\documentclass[12pt,a4paper,twoside]{article}
 
\pdfoutput=1
\usepackage[utf8]{inputenc}
\usepackage[T1]{fontenc} 
\usepackage{RR}

\usepackage{graphicx} 
\usepackage{mdwlist}
\usepackage{mathtools}
\usepackage{amsmath}
\usepackage{amssymb,amsmath,stmaryrd,ifsym,amsthm}
\usepackage{verbatim}
\usepackage[defblank]{paralist}
\usepackage{xspace} 
 
\usepackage[numbers,square,sort&compress]{natbib}
\usepackage{algorithmicx}
\usepackage[noend]{algpseudocode}
\usepackage{algorithm}
\usepackage[small,it]{caption}
 
\usepackage{lscape}
\usepackage{color} 
 
\usepackage{tikz}
\usepgflibrary{shapes.symbols} 
\usetikzlibrary{shapes.symbols} 
\usepackage{pstricks}
\usepackage{pst-node}   
\usepackage{pst-grad}   
\usepackage{pst-plot}   

\usepackage{macros}
\usepackage{macros_database}

\newcommand{\NMSILongName}{Non-Monotonic Snapshot Isolation}

\newenvironment{IEEEproof}{\begin{proof}}{\end{proof}}

\begin{document}
%
\RRtitle{Non-Monotonic \\ Snapshot Isolation}
\RRetitle{Non-Monotonic \\ Snapshot Isolation%
  \thanks{%
    The work presented in this paper has been partially funded
    by ANR projects Prose (ANR-09-VERS-007-02) and Concordant (ANR-10-BLAN 0208).}}
\titlehead{\NMSILongName} 


  
\RRauthor{    
      Masoud {Saeida Ardekani} {\footnotesize{UPMC-LIP6}}\\
      Pierre Sutra  {\footnotesize{University of Neuch\^atel}} \\
      Nuno Pregui\c{c}a {\footnotesize{Universidade Nova de Lisboa}} \\ 
      Marc Shapiro {\footnotesize{INRIA \& UPMC-LIP6}} \\}

\authorhead{Saeida Ardekani, Sutra, Pregui\c{c}a, Shapiro}

  \RRdate{November 2011}

\RRversion{5}

\RRNo{7805}
\RRabstract{  
\footnotesize     
  Many distributed applications require transactions.
  However, transactional protocols that require strong synchronization
  are costly in large scale environments.
  Two properties help with scalability of a transactional system:
  genuine partial replication (\GPR), which leverages
  the intrinsic parallelism of a workload, 
  and snapshot isolation (\SI), which decreases the need for synchronization.
  We show that, under standard assumptions (data store accesses are not
  known in advance, and transactions may access arbitrary objects in the
  data store), it is impossible to have both SI and GPR. 
  To circumvent this impossibility, we propose a weaker consistency
  criterion, called \NMSILongName~(\NMSI).
  \NMSI retains the most important properties of \SI, i.e., read-only
  transactions always commit, and two write-conflicting updates do not both
  commit.
  We present a \GPR protocol that ensures \NMSI, and has lower message cost
  (i.e., it contacts fewer replicas and/or commits faster) 
  than previous approaches.
}

\RRkeyword{distributed systems; transcational systems; replication; concurrency control; transactions; database}

\RRresume{ Cet article étudie deux propriétés favorisant le passage
 à l'échelle des systèmes répartis transactionnels: 
 la réplication partielle authentique (GPR), et le critère de cohérence Snapshot Isolation (SI).
 GPR spécifie que pour valider une transaction T, seules
 les répliques des données accédées par T effectuent des pas de calcul.
 SI définit que toute transaction doit lire une vue cohérente du système,
 et que deux transactions concurrentes ne peuvent écrire la même donnée.
 Nous montrons que SI et GPR sont deux propriétés incompatibles.
 Afin de contourner cette limitation, 
 nous proposons un nouveau critère de cohérence: Non-Monotonic Snapshot Isolation (NMSI).
 NMSI est proche de SI et néanmoins compatible avec GPR.
 Afin de justifier ce dernier point, nous présentons un protocole authentique implémentant de manière efficace NMSI.
 Au regard des travaux précédents sur le contrôle de concurrence dans les systèmes répartis transactionnels, 
 notre protocole est le plus performant en latence et/ou en nombre de messages échangés.
}

\RRmotcle{systèmes répartis, systèmes transactionnels, contrôle de concurrence, transaction, base de données} 

\RRprojets{Regal}
\RCParis

\makeRR
\makeRT
 
\section{Introduction}
\labsection{introduction}

Large scale transactional systems have conflicting requirements.
On the one hand, strong transactional guarantees are fundamental to
many applications.
On the other, remote communication and synchronization is costly and
should be avoided.%
\footnote{
  We address general-purpose transactions, i.e., we assume that a
  transaction may access any object in the system, and that its read- and
  write-sets are not known in advance.
}

To maintain strong consistency guarantees while alleviating the high cost of synchronization,
 Snapshot Isolation (\SI)
 is a popular approach in both distributed database replications 
\cite{Serrano2007,Armendariz-Inigo2008,Daudjee2006},
 and software transactional memories \cite{Bieniusa2010,Riegel2006}.
Under \SI, a transaction accesses its own \emph{consistent snapshot} of the
data, which is unaffected by concurrent updates.%
~A read-only transaction always commits unilaterally and without
synchronization.
An update transaction synchronizes on commit to ensure that no
concurrent conflicting transaction has committed before it.

Our first contribution is to prove that \SI is equivalent to the
conjunction of the following properties:
\begin{inparaenum}[\em (i)]
\item
  no cascading aborts, 
\item
  strictly consistent snapshots, i.e., a transaction observes a
  snapshot that coincides with some point in (linear) time,
\item
  two concurrent write-conflicting update transactions never both commit, and
\item
  snapshots observed by transactions are monotonically ordered.
\end{inparaenum}
Previous definitions \cite{Adya99,Elnikety2005} of \SI
extend histories with abstract snapshot points.
Our decomposition shows that \SI can be expressed on plain histories 
like serializability \cite{Berenson1995}.

Modern data stores replicate data for both performance and availability.
Full replication does not scale, as every process must perform all updates.
\emph{Partial replication} (PR) aims to address this problem, by replicating
only a subset of the data at each process.
Thus, if transactions would communicate only over the minimal number of
replicas, synchronisation and computation overhead would be reduced.
However, in the general case, the overlap of transactions cannot be
predicted; therefore, many PR protocols perform system-wide 
global consensus \cite{Serrano2007,Armendariz-Inigo2008} or
communication \cite{Sovran2011}.
This negates the potential advantages of PR; hence, we require
\emph{genuine} partial replication \cite{Schiper2010} (\GPR), in which a
transaction communicates only with those processes that replicate some
object accessed in the transaction.
With \GPR, independent transactions do not interfere with each other,
and the intrinsic parallelism of a workload can be exploited.
Our second contribution is to show that \SI and \GPR are 
incompatible. More precisely, we prove that an asynchronous message-passing system supporting \GPR
cannot compute monotonically ordered snapshots, nor strictly consistent ones, even if it is failure-free.

The good news is our third contribution: a consistency criterion,
called \emph{\NMSILongName} (\NMSI) that overcomes this impossibility.
\NMSI is very similar to \SI, as
\begin{inparablank}
\item
  every transaction observes a consistent snapshot, 
  and
\item
  two concurrent
  write-conflicting updates never both commit.
\end{inparablank}
However, under \NMSI, snapshots are neither strictly consistent nor monotonically ordered.

Our final contribution is a \GPR protocol ensuring \NMSI, called \jessy.
\jessy uses a novel variant of version vectors, called \emph{dependence
  vectors}, to compute consistent partial snapshots asynchronously.
\noindent
To commit an update transaction, \jessy uses a single atomic multicast.
Compared to previous protocols, \jessy commits transactions faster
and/or contacts fewer replicas.

This paper proceeds as follows.
We introduce our system model in \refsection{model}. 
\refsection{reconstruct} presents our decomposition of \SI.
\refsection{imp} shows that \GPR and \SI are mutually incompatible. 
We introduce \NMSI in \refsection{wsi}. 
\refsection{protocol} describes \jessy, our \NMSI protocol.
We compare with related work in \refsection{relatedwork}, and conclude
in \refsection{conclusion}.

\newpage
\section{Model}
\labsection{model}

This section defines the elements in our model and formalizes \SI and \GPR\@.

\subsection{Objects \& transactions}
\labsection{model:base}

Let \objectSet be a set of objects, and \transSet be a set of transaction identifiers.
Given an object $x$ and an identifier $i$, $x_i$ denotes \emph{version} $i$ of $x$. 
A \emph{transaction} $T_{i \in \transSet}$ is a finite permutation of
read and write operations followed by a \emph{terminating} operation,
commit ($c_i$) or abort ($a_i$).
We use $w_i(x_i)$ to denote transaction $T_i$ writing version $i$ of object $x$, 
and $r_i(x_j)$ to mean that $T_i$ reads version $j$ of object $x$.
In a transaction,
every write is preceded by a read on the same object,
and every object is read or written at most once.%
\footnote{
  These restrictions ease the exposition of our results but do not change their validity.
}
We note \writeSetOf{T_i} the write set of
$T_i$, i.e., the set of objects written by transaction $T_i$.
Similarly, \readSetOf{T_i} denotes the read set of transaction $T_i$. 
The \emph{snapshot} of $T_i$ is the set of versions read by $T_i$.
Two transactions \emph{conflict} when they access the same object and one of them modifies it;
they \emph{write-conflict} when they both write to the same object.

\subsection{Histories}
\labsection{model:hist}

A \emph{complete history} $h$ is a partially ordered set of operations such that
(1) for every operation $o_i$ appearing in $h$, transaction $T_i$ terminates in $h$,
(2) for every two operations $o_i$ and $o_i'$ appearing in $h$,
    if $o_i$ precedes $o_i'$ in $T_i$, then $o_i <_h o_i'$,
(3) for every read $r_i(x_j)$ in $h$, 
    there exists a write operation $w_j(x_j)$ such that $w_j(x_j) <_h r_i(x_j)$,
and 
(4) any two write operations over the same objects are ordered by $<_h$.
A \emph{history} is a prefix of a complete history.
For some history $h$, 
order $<_{h}$ is the \emph{real-time order} induced by $h$.
Transaction $T_i$ is \emph{pending} in history $h$ if $T_i$ does not commit, nor abort in $h$.
We note $\versionOrder_h$ the version order induced by $h$ between different versions
of an object, i.e., for every object $x$, and every pair of transactions $(T_i,T_j)$,
$x_i \versionOrder_h x_j \iff w_i(x_i) <_h w_j(x_j)$.
Following Bernstein et al. \cite{Bernstein1987}, we depict a history as a graph.
We illustrate this with history $h_1$ below
in which transaction $T_a$ reads the initial versions of objects $x$ and $y$,
while transaction $T_1$ (respectively $T_2$) updates $x$ (resp. $y$).%
\footnote{
  Throughout the paper, read-only transactions are specified with
  an alphabet subscript, and update transactions are shown with numeric
  subscript.
}
\begin{figure}[!ht]
  \centering
  \begin{tikzpicture}
    \draw (0,2) node[left]{\small $h_1=$};
    \draw (1,2) node[left]{\small $r_a(x_0)$};
    \draw [->] (1,2) -- (1.5,2);
    \draw (1.5,2) node[right]{\small $r_1(x_0).w_1(x_1).c_1$};
    \draw (3.5,1.3) node[left]{\small $r_a(y_0).c_a$};
    \draw [->] (3.5,1.3) -- (4,1.3);
    \draw (4,1.3) node[right]{\small $r_2(y_0).w_2(y_2).c_2$};

    \draw [->] (0.6,1.8) -- (2,1.3);

  \end{tikzpicture}
\end{figure}

\noindent%
When order $<_h$ is total, 
we shall write a history as a permutation of operations, e.g., $h_2=r_1(x_0).r_2(y_0).w_2(y_2).c_1.c_2$.

\subsection{Snapshot Isolation} 
\labsection{model:si}

Snapshot isolation (\SI) was introduced by Berenson et al. \cite{Berenson1995}, then
later generalized under the name GSI by Elnikety et al. \cite{Elnikety2005}.
In this paper, we make no distinction between \SI and \GSI.

Let us consider a function $\mathcal{S}$ which takes as
input a history $h$, and returns an extended history $h_s$ by adding a
\emph{snapshot point} to $h$ for each transaction in $h$.
Given a transaction $T_i$, the snapshot point of $T_i$ in $h_s$, denoted
$s_i$, precedes every operation of transaction $T_i$ in $h_{s}$.
A history $h$ is in \SI if, and only if, there exists a function
$\mathcal{S}$ such that $h_s=\mathcal{S}(h)$ and $h_s$ satisfies the following rules:

\vspace{1em}
\begin{tabular}{@{}l@{~~~~~~~~~~~~}r@{}}
  \begin{minipage}{6cm}
    \textbf{D1 (Read Rule)}\\
    $
    \forall r_{i}(x_{j \neq i}),  w_{k \neq j}(x_{k}), c_{k}  \in h_{s}: \\
    ~~~~~~   c_{j} \in h_{s} \hspace*{\fill} (D1.1) \\
    ~~\land~ c_{j} <_{h_s} s_{i} \hspace*{\fill} (D1.2) \\
    ~~\land~ (c_{k} <_{h_s} c_{j} \lor  s_{i} <_{h_s} c_{k}) \hspace*{\fill} (D1.3)
    $
  \end{minipage}
  &
  \begin{minipage}{5cm}
  \vspace{-1.25em}%
   \textbf{D2 (Write Rule)}\\
   $
   \forall c_{i}, c_{j} \in h_{s}: \\
   ~~~ \writeSetOf{T_i} \inter \writeSetOf{T_j} \neq \emptySet \\
   ~~~ \implies \left(  c_i <_{h_s} s_j \lor  c_j <_{h_s} s_i \right)
   $
  \end{minipage}
\end{tabular}

\subsection{System}
\labsection{model:system}

We consider a message-passing system of $n$ processes $\procSet = \{p_1, \ldots , p_n\}$.
Links are quasi-reliable.
We shall define our synchrony assumptions later.
Following Fischer et al. \cite{Fischer1985}, 
an execution is a sequence of steps made by one or more processes.
During an execution, processes may fail by crashing.
A process that does not crash is said \emph{correct}; otherwise it is \emph{faulty}.
We note $\mathfrak{F}$ the refinement mapping
\cite{AbadiLamport-Theexistenceofrefin} from executions to histories, i.e.,
if \run is an execution of the system, then $\refMapOf{\run}$ is the history produced by $\run$.
A history $h$ is \emph{acceptable} if there exists an execution $\run$
such that $h=\mathfrak{F}(\run)$.
We consider that given two sequences of steps $U$ and $V$,
if $U$ precedes $V$ in some execution $\run$,
then the operations implemented by $U$ precedes (in the sense of $<_h$) 
the operations implemented by $V$ in the history $\mathfrak{F}(\run)$.%
\footnote{
  Notice that since steps to implement operations may interleave,
  $<_h$ is not necessarily a total order.
}

\subsection{Partial Replication}
\labsection{model:par}

A data store $\mathcal{D}$ is a finite set of tuples $(x, v, i)$ where
$x$ is an object (data item), $v$ a value, and $i \in \transSet$ a version.
Each process in $\procSet$ holds a data store such that initially every object $x$ has version $x_0$.
For an object $x$, $\replicaSetOf{x}$ denotes the set of processes, or \emph{replicas}, that hold a copy of $x$.
By extension for some set of objects $X$, $\replicaSetOf{X}$ denotes the replicas of $X$;
given a transaction $T_i$, $\replicaSetOf{T_i}$ equals $\replicaSetOf{\readSetOf{T_i} \union \writeSetOf{T_i}}$.

We make no assumption about how objects are replicated.
The coordinator of $T_i$, denoted $\dbCoordinatorOf{T_i}$, 
is in charge of executing $T_i$ on behalf of some client (not modeled). 
The coordinator does not know in advance the read set or the write set of $T_i$.
To model this, we consider that every prefix of a transaction (followed by a terminating operation)
is a transaction with the same id.
 
Genuine Partial Replication (\GPR) aims to ensure that, when the workload is parallel, 
throughput scales linearly with the number of nodes \cite{Schiper2010}:
\begin{itemize}
\item \textbf{\GPR.}
  For any transaction $T_i$, only processes that replicate objects accessed
  by $T_i$ make steps to execute $T_i$.
\end{itemize}

\subsection{Progress}
\labsection{model:progress}

The read rule of \SI does not define what is the snapshot to be read.
According to Adya \cite{Adya99}, 
``transaction $T_i$'s snapshot point needs not be chosen after the most recent commit when $T_i$ started, 
but can be selected to be some (convenient) earlier point.''
As a consequence, \SI does not preclude a transaction to always observe outdated data.
This implies that an update transaction may always abort even if it runs alone.
To ensure that a transactional system remains practical, 
\citet{Herlihy:2003}, as well as \citet{Guerraoui:2009}, consider that an update transaction should abort only if a conflict occurs.
In the case of \SI, we require that this property holds for write-conflict, i.e., 
a query never forces an update to abort.

\begin{itemize}
\item \textbf{Obstruction-free Updates (\OFU).} 
  For every update transaction $T_i$,  
  if \coordOf{T_i} is correct then $T_i$ eventually terminates.
  Moreover, if $T_i$ does not \emph{write-conflict} with some concurrent transaction then $T_i$ eventually commits.
\end{itemize}

Most workloads exhibit a high proportion of read-only transactions, or \emph{queries}.
The wait-free queries property (see below) ensures that such accesses are fast.
\SI was designed at core to offer this property.

\begin{itemize}
\item \textbf{Wait-free Queries (\WFQ).}
  A read-only transaction $T_i$ never waits for another transaction and eventually commits.
\end{itemize}

\newpage
\section{Decomposing \SI}
\labsection{reconstruct}
 
This section defines four properties, whose conjunction is
necessary and sufficient to attain \SI.
We later use these properties in \refsection{imp} to derive our impossibility result.

\subsection{Cascading Aborts}
\labsection{reconstruct:aca}

Intuitively, a read-only transaction must abort 
if it observes the effects of an uncommitted transaction that later aborts.
By guaranteeing that every version read by a transaction is committed,
rules D1.1 and D1.2 of \SI prevent such a situation to occur.
In other words, these rules \emph{avoid cascading aborts}.
We formalize this property below:
\begin{definition}[Avoiding Cascading aborts]
  History $h$ avoids cascading aborts, if for every read $r_i(x_j)$ in $h$,
  $c_j$ precedes $r_i(x_j)$ in $h$.
  $\ACA$ denotes the set of histories that avoid cascading aborts.
\end{definition}

\subsection{Consistent and Strictly Consistent Snapshots}
\labsection{reconstruct:scons}

Consistent and strictly consistent snapshots are defined by refining causality into a dependency relation
as follows:

\begin{definition}[Dependency]
  Consider a history $h$ and two transactions $T_i$ and $T_j$.
  We note $T_i \MustHave T_j$ when $r_i(x_j)$ is in $h$.
  Transaction $T_i$ depends on transaction $T_j$ when $T_i \depend T_j$ holds.%
  \footnote{
    We note $\mathcal{R}^{*}$ the transitive closure of some binary relation $\mathcal{R}$.
  }
  Transaction $T_i$ and $T_j$ are independent if neither $T_i \depend T_j$, nor $T_j \depend T_i$ hold.
\end{definition}

This means that a transaction $T_i$ depends on a transaction $T_j$ 
if $T_i$ reads an object modified by $T_j$, or such a relation holds by transitive closure.
To illustrate this definition, consider history $h_3=r_1(x_0).w_1(x_1).c_1.r_a(x_1).c_a.r_b(y_0).c_b$.
In $h_3$, transaction $T_a$ depends on $T_1$.Ho
Notice that, even if $T_1$ causally precedes $T_b$, $T_b$ does not depend on $T_1$ in $h_3$.

We now define consistent snapshots with the above dependency relation.
A transaction sees a consistent snapshot iff it observes the effects of all transactions
it depends on \cite{Chan1985}.
For example, consider the history $h_4=r_1(x_0).w_1(x_1).c_1.r_2(x_1).r_2(y_0).w_2(y_2).c_2.r_a(y_2).r_a(x_0).c_a$
In this history, transaction $T_a$ does not see a consistent snapshot: 
$T_a$  depends on $T_2$, and $T_2$ also depends on $T_1$, 
but $T_a$ does not observe the effect of $T_1$ (i.e., $x_1$). 
Formally, consistent snapshots are defined as follows:
\begin{definition}[Consistent snapshot]
  A transaction $T_i$ in a history $h$ observes a consistent snapshot
  iff, for every object $x$, 
  if 
  (i) $T_i$ reads version $x_j$,
  (ii) $T_k$ writes version $x_k$,
  and (iii) $T_i$ depends on $T_{k}$, 
  then version  $x_k$ is followed by version $x_j$ in the version order induced by $h$ ($x_k \ll_h x_j$).
  We write $h \in \CONS$ when all transactions in $h$ observe a
  consistent snapshot.
\end{definition}

\SI requires that a transaction observes the
committed state of the data at some \emph{point} in the past.
This requirement is stronger than consistent snapshot.
For some transaction $T_i$, it implies that 
\emph{(i)} there exists a snapshot point for $T_i$ (\SCONSa), 
and \emph{(ii)} if transaction $T_i$ observes the effects of transaction $T_j$, 
it must also observe the effects of all transactions that precede $T_j$ in time (\SCONSb).
A history is called strictly consistent if both \SCONSa and \SCONSb hold.
For instance, consider the following history: 
$h_5=r_1(x_0).w_1(x_1).c_1.r_a(x_1).r_2(y_0).w_2(y_2).c_2.r_a(y_2).c_a$.
Because $r_a(x_1)$ precedes $c_2$ in $h_5$, $y_2$ cannot be observed when $T_a$ takes its snapshot.
As a consequence, the snapshot of transaction $T_a$ is not strictly consistent.
This issue is disallowed by \SCONSa. 
Now, consider history $h_6=r_1(x_0).w_1(x_1).c_1.r_2(y_0).w_2(y_2).c_2.r_a(x_0).r_a(y_2).c_a$.
Since $c_1$ precedes $c_2$ in $h_6$ and transaction $T_a$ observes the effect of $T_2$ (i.e., $y_2$),
it should also observe the effect of $T_1$ (i.e., $x_1$).
\SCONSb prevents history $h_6$ to occur.

\begin{definition}[Strictly consistent snapshot]
  Snapshots in history $h$ are strictly consistent, when 
  for any committed transactions $T_i$, $T_j$, $T_{k \neq j}$ and $T_{l}$,
  the following two properties hold:
  \begin{itemize}
    \item[-] $\forall r_i(x_j), r_i(y_l) \in h : r_i(x_j) \not <_h  c_l$ \hspace*{\fill} $(\SCONSa)$
    \item[-] $\forall r_i(x_j), r_i(y_l), w_k(x_k) \in h : $ \\
             $ ~~~~~~~~~~~~~~~ c_k <_{h} c_l \implies c_k <_h c_j$ \hspace*{\fill} $(\SCONSb)$
  \end{itemize}
  We note \SCONS the set of strictly consistent histories.
\end{definition}


\subsection{Snapshot Monotonicity}
\labsection{reconstruct:mon}

In addition, \SI requires what we call monotonic snapshots.
For instance, although history $h_7$ below satisfies \SCONS, this history does not belong to \SI:
since $T_a$ reads $\{x_0, y_2\}$, and $T_b$ reads $\{x_1, y_0\}$,
there is no extended history that would guarantee the read rule of \SI.

\begin{figure}[!h]
  \vspace{-1em}
  \centering
  \begin{tikzpicture}
    \draw (0,2) node[left]{\small $h_7=$}; 
    \draw (1,2) node[left]{\small $r_a(x_0)$};
    \draw [->] (1,2) -- (1.5,2);
    \draw (1.5,2) node[right]{\small $r_1(x_0).w_1(x_1).c_1$};
    \draw [->] (4.4,2) -- (5,2);
    \draw (5,2) node[right]{\small $r_b(x_1).c_b$};
    \draw (1,1.3) node[left]{\small $r_b(y_0)$};
    \draw [->] (1,1.3) -- (1.5,1.3);
    \draw (1.5,1.3) node[right]{\small $r_2(y_0).w_2(y_2).c_2$};
    \draw [->] (4.4,1.3) -- (5,1.3);
    \draw (5,1.3) node[right]{\small $r_a(y_2).c_a$};

    \draw [->] (1,1.9) -- (5,1.4);
    \draw [->] (1,1.4) -- (5,1.9);

  \end{tikzpicture}
  \vspace{-1em}
\end{figure}


\SI requires monotonic snapshots.
However, the underlying reason is intricate enough that some previous
works \cite[for instance]{Bieniusa2010} do not ensure this property, while claiming to be \SI.
Below, we introduce an ordering relation between snapshots to formalize snapshot monotonicity.

\begin{definition}[Snapshot precedence]
  Consider a history $h$ and two distinct transactions $T_i$ and $T_j$.
  The snapshot read by $T_i$ precedes the snapshot read by $T_j$ in history $h$, written $T_i \rightarrow T_j$,
  when
  $r_i(x_k)$ and $r_j(y_l)$ belong to $h$
  and
  either (i) $r_i(x_k) <_h c_l$ holds,
  or (ii) transaction $T_l$ writes $x$ and $c_k <_h c_l$ holds.
\end{definition}

\noindent
For more illustration, consider
$h_8=r_1(x_0).w_1(x_1).c_1.r_2(y_0).w_2(y_2).r_a(x_1).c_2\\.r_b(y_2).c_a.c_b$
and $h_9=r_1(x_0).w_1(x_1).c_1.r_a(x_1).c_a.r_2(x_1).r_2(y_0).w_2(x_2).w_2(y_2).c_2\\.r_b(y_2).c_b$.
In history $h_8$, $T_a \rightarrow T_b$ holds because $r_a(x_1)$ precedes $c_2$ and $T_b$ reads $y_2$.
In $h_9$, $c_1$ precedes $c_2$ and both $T_1$ and $T_2$ modify object $x$.
Thus, $T_a \rightarrow T_b$ also holds.
We define snapshot monotonicity using snapshot precedence as follows:

\begin{definition}[Snapshot monotonicity]
  Given some history $h$, 
  if the relation \precedes induced by $h$ is a partial order,
  the snapshots in $h$ are \emph{monotonic}.
  We note \MON the set of histories that satisfy this property.
\end{definition}

According to this definition, since both $T_a \rightarrow T_b$ and $T_b \rightarrow T_a$ hold in history $h_7$,
this history does not belong to \MON.

Non-monotonic snapshots are observed under update serializability \cite{GM80},
that is when queries observe consistent state, but only updates are serializable.

\subsection{Write-Conflict Freedom}
\labsection{reconstruct:wcf}

Rule D2 of \SI forbids two concurrent write-conflicting transactions from both committing.
Since in our model we assume that every write is preceeded by a corresponding read on the same object, 
every update transaction depends on a previous update transaction (or on the initial transaction $T_0$).
Therefore, under \SI, concurrent conflicting transactions must be independent:
 
\begin{definition}[Write-Conflict Freedom]
  A history $h$ is write-conflict free if two independent transactions never write to the same object.
  We denote by $\WCF$ the histories that satisfy this property.
\end{definition}

\subsection{The decomposition}
\labsection{reconstruct:si}

\reftheo{reconstruct:main} below establishes that a history $h$ is in \SI iff
(1) every transaction in $h$ sees a committed state,
(2) every transaction in $h$ observes a strictly consistent snapshot,
(3) snapshots are monotonic,
and (4) $h$ is write-conflict free.

\begin{lemma}
  \lablem{reconstruct:1}
  Consider a history $h \in \SI$ and two versions $x_i$ and $x_j$ of some object $x$.
  If $x_i \versionOrder_{h} x_j$ holds then $T_j \depend T_i$ is true.
\end{lemma}
 
 \begin{proof}
  Assume some history $h \in \SI$ such that $x_i \versionOrder_h x_j$ holds.
  Let $h_s$ be an extended history for $h$ that satisfies rules D1 and D2.
  According to the model, transaction $T_j$ first reads some version $x_k$, then writes version $x_j$.

  First, assume that there is no write to $x$ between $w_i(x_i)$ and $w_j(x_j)$.
  Since $x$ belongs to $\writeSetOf{T_i} \inter \writeSetOf{T_j}$,
  rule D2 tells us that either $c_i <_{h_s} s_j$, or  $c_j <_{h_s} s_i$ holds.
  We observe that because $x_i \versionOrder_h x_j$ holds, it must be true that $c_i <_{h_s} s_j$.
  Since there is no write to $x$ between $w_i(x_i)$ and $w_j(x_j)$, $x_k \versionOrder x_i$ holds, or $k=i$.
  Observe that in the former case rule D1.3 is violated.
  Thus, transaction $T_j$ reads version $x_i$.
  To obtain the general case, we apply inductively the previous reasoning.
\end{proof}

\begin{lemma}
  \lablem{reconstruct:2}
  Let $h \in \SI$ be a history, 
  and $\mathcal{S}$ be a function such that $h_s=\mathcal{S}(h)$ satisfies D1 and D2.
  Consider $T_i, T_j \in h$.
  If $T_i \rightarrow T_j$ holds then $s_i <_{h_s} s_j$.
\end{lemma}

\begin{proof}
  Consider two transactions $T_i$ and $T_j$ such that the snapshot of $T_i$ precedes the snapshot of  $T_j$.
  By definition of the snapshot precedence relation, 
  there exist $T_k, T_l \in h$ such that $r_i(x_k), r_j(y_l) \in h$ and
  either \emph{(i)} $r_i(x_k) <_h c_l$\ , 
  or \emph{(ii)} $w_l(x_l) \in h$ and $c_k <_h c_l$.
  Let us distinguish each case:
  \begin{compactitem}
  \item[(Case $r_i(x_k) <_h c_l$)]
    By definition of function $\mathcal{S}$, $s_i$ precedes $r_i(x_k)$ in $h_s$.
    From $r_j(y_l) \in h$ and rule D1.2, $c_l <_{h_s} s_j$ holds.
    Hence, $s_i <_{h_s} s_j$ holds.
  \item[(Case $c_k <_h c_l$)]
    From \emph{(i)} $r_i(x_k), w_l(x_l) \in h$, \emph{(ii)} $c_k <_h c_l$ and \emph{(iii)} rule D1.3, 
    we obtain $s_i <_{h_s} c_l$.
    From $r_j(y_l) \in h$ and rule D1.2, $c_l <_{h_s} s_j$ holds.
    It follows that $s_i <_{h_s} s_j$ holds.
  \end{compactitem}
\end{proof}

\begin{lemma}
  \lablem{reconstruct:3}
  Consider a history $h \in \ACA \inter \CONS \inter \WCF$, and two versions $x_i$ and $x_j$ of some object $x$. 
  If $x_i \versionOrder_h x_j$ holds then $c_i <_h c_j$.
\end{lemma}

\begin{proof}
  Since both $T_i$ and $T_j$ write to $x$ and $h$ belongs to $\WCF$
  either $T_j \depend T_i$ or $T_i \depend T_j$ holds.
  We distinguish the two cases below:
  \begin{compactitem}
  \item[(Case $T_j \depend T_i$)]
    First, assume that $T_j \MustHave T_i$ holds.
    Note $y$ an object such that $r_j(y_i)$ is in $h$.
    Since $h$ belongs to $\ACA$, $c_i <_h r_j(y_i)$ holds.
    Because $h$ is an history, $r_j(y_i) <_h c_j$ must hold.
    Hence we obtain $c_i <_h c_j$.
    By a short induction, we obtain the general case. 
  \item[(Case $T_i \depend T_j$)]
    Let us note $x_k$ the version of $x$ read by transaction $T_i$.
    From the definition of an history and since $h$ belongs to to \ACA, we know that
    $w_k(x_k) <_h c_k <_h r_i(x_k) <_h w_i(x_i)$ holds.
    As a consequence, $x_k \versionOrder_h x_i$ is true.
    Since
    \emph{(i)} $h$ belongs to \CONS, 
    \emph{(ii)} $T_i \depend T_j$,
    and \emph{(iii)} $T_j$ writes to $x$,
    it must be the case that $x_j \versionOrder_h x_k$.
    We deduce that $x_j \versionOrder_h x_i$ holds; a contradiction.
  \end{compactitem}
\end{proof}

\noindent
Using these lemmata, we successively prove each inclusion.

\begin{proposition}
  \labprop{reconstruct:1}
  $\SI \subseteq \ACA \inter \SCONS \inter \WCF \inter \MON$
\end{proposition}

\begin{proof}
  Choose $h$ in \SI.
  Note $\mathcal{S}$ a function such that history $h_s=\mathcal{S}(h)$ satisfies rules D1 and D2.

  \begin{compactitem}

  \item[($h \in \ACA$)]
    It is immediate from rules D1.1 and D1.2. 

  \item[($h \in \WCF$)]
    Consider two independent transactions $T_i$ and $T_j$ modifying the same object $x$.
    By the definition of a history, $x_i \versionOrder_h x_j$ \ , or $x_j \versionOrder_h x_i$ holds.
    Applying \reflem{reconstruct:1}, we conclude that in the former case $T_j$ depends on $T_i$,
    and that the converse holds in the later.

  \item[($h \in \SCONSa$)]
    By contradiction.
    Assume three transactions $T_i$, $T_j$ and $T_l$ such that
    $r_i(x_j), r_i(y_l) \in h$ and $r_i(x_j) <_h c_l$ are true.
    In $h_s$, the snapshot point $s_i$ of transaction $T_i$ is placed prior to every operation of $T_i$ in $h_s$.
    Hence, $s_i$ precedes $r_i(x_j)$ in $h_s$.
    This implies that  $s_i <_{h_s} c_l \land r_i(y_l) \in h_s$ holds.
    A contradiction to rule D1.2.
    
  \item[($h \in \SCONSb$)]
    Assume for the sake of contradiction four transactions $T_i$, $T_j$, $T_{k \neq j}$ and $T_l$
    such that: $r_i(x_j), r_i(y_l)$,  $w_k(x_k) \in h$, $c_k <_h c_l$ and $c_k \not<_{h} c_j$ are all true.
    Since transaction $T_j$ and $T_k$ both write $x$, by rule D2, we know that $c_j <_{h_s} c_k$ holds.
    Thus, $c_j <_{h_s} c_k <_{h_s} c_l$ holds.
    According to rule D1.2, since $r_i(y_l)$ is in $h$, $c_l <_{h_s} s_i$ is true.
    We consequently obtain that $c_j <_{h_s} c_k < s_i$ holds.
    A contradiction to rule D1.3.

  \item[($h \in MON$)]
    If \precedes is not a  partial order, there exist transactions $T_1, \ldots, T_{n \geq 1}$ 
    such that: $T_1 \rightarrow \ldots \rightarrow T_n \rightarrow T_1$.
    Applying \reflem{reconstruct:2}, we obtain that the relation $s_1 <_{h_s} s_1$ is true.
    A contradiction.

  \end{compactitem}
\end{proof}

\begin{proposition}
  \labprop{reconstruct:2}
  $\ACA \inter \SCONS \inter \WCF \inter \MON \subseteq \SI$
\end{proposition}

\begin{proof}
  Consider some history $h$ in $\ACA \inter \SCONS \inter \WCF \inter \MON$.
  If history $h$ belongs to $\SI$ 
  then there must exist a function $\mathcal{S}$ such that  $h'=\mathcal{S}(h)$ satisfies rules D1 and D2. 
  In what follows, we build such an extended history $h'$, then we prove its correctness.

  \vspace{0.5em}
  \noindent
  [Construction]
  Initially $h'$ equals $h$.
  For every transaction $T_i$ in $h'$ we add a snapshot point $s_i$ in $h'$,
  and for every operation $o_i$ in $h'$, we execute the following steps:
  \begin{compactitem}
  \item[\textbf{S1.}] We add the order $(s_i,o_i)$ to $h'$.
  \item[\textbf{S2.}] If $o_i$ equals $r_i(x_j)$ for some object $x$ then
    \begin{compactitem}
    \item[\textbf{S2a.}] we add the order $(c_j,s_i)$ to $h'$,
    \item[\textbf{S2b.}] and, for every committed transaction $T_k$ such that $w_k(x_k)$ is in $h$,
          if $c_k <_{h} c_j$ does not hold then we add the order $(s_i,c_k)$ to $h'$.
    \end{compactitem}
  \end{compactitem}
  
  \vspace{0.5em}
  \noindent
  [Correctness]
  We now prove that $h'$ is an extended history that satisfies rules D1 and D2.
  \begin{itemize}

  \item $h'$ is an extended history.

    Observe that for every transaction $T_i$ in $h'$, there exists a snapshot point $s_i$,
    and that according to step S1, $s_i$ is before all operations of transaction $T_i$.
    It remains to show that order $<_{h'}$ is acyclic.
    We proceed by contradiction.

    Since $h$ is a history, it follows that any cycle formed by relation $<_{h'}$  contains a snapshot point $s_i$.
    Furthermore, according to steps S1 and S2 above, we know that for some operation $c_{j \neq i}$,
    relation $c_j <_{h'} s_i <_{h'}^* c_j$ holds.

    By developing relation $s_i <_{h'}^* c_j$, we obtain the following three relations.
    The first two relations are terminal, while the last is recursive.
    \begin{compactitem}
    \item Relation $s_i <_{h'} c_j$ holds.
      This relation has to be produced by step S2b.
      Hence, there exist operations $r_i(x_k), w_j(x_j)$ in $h'$ such that $c_j <_h c_k$ does not hold.
      Observe that since $h$ belongs to $\ACA \inter \CONS \inter \WCF$, 
      by \reflem{reconstruct:3}, it must be the case that  $c_k <_h c_j$ holds.
    \item Relation $s_i <_{h'} o_i <_h^* c_j$ holds for some read operation $o_i$ in $T_i$.
      (If $o_i <_h^* c_j$ with $o_i$ a write or a terminating operation, 
      we may consider a preceding read that satisfies the same relation.)
    \item Relation $s_i <_{h'} o_i <_{h'}^* c_j$ holds for some read operation $o_i$ in $T_i$,
      and $o_i <_{h'}^* c_j$ does not imply  $o_i <_h^* c_j$.
      (Again if $o_i$ is a write or a terminating operation, we may consider a preceding read that satisfies this relation.)
      Relation $o_i <_{h'}^* c_j$ cannot be produced by steps S1 and S2.
      Hence, there must exist a commit operation $c_k$ and a snapshot point $s_l$ such that
      $s_i <_{h'} o_i <_h c_k <_{h'} s_l <_{h'}^* c_j$ holds.
    \end{compactitem}
    From the result above, we deduce that there exist
    snapshot points $s_1, \ldots, s_{n \geq 1}$ and commit points $c_{k_1} \ldots c_{k_n}$ such that:
    \begin{equation}
      \labequation{reconstruct:1}
      s_1 \hb c_{k_1} <_{h'} s_2 \hb c_{k_2} \ldots s_n \hb c_{k_n} <_{h'} s_1
    \end{equation}
    where $s_i \hb c_{k_i}$ is a shorthand for 
    either \emph{(i)} $s_i <_{h'} c_{k_i}$ with $r_i(x_j), w_{k_i}(x_{k_i}) \in h$ and $c_j <_h c_{k_i}$,
    or \emph{(ii)} $s_i <_{h'} o_i <_h c_{k_i}$ with $o_i$ is some read operation.

    We now prove that for every $i$,  $T_i \rightarrow T_{i+1}$ holds.
    Consider some $i$.
    First of all, observe that a relation $c_{k_{i-1}} < s_i$ is always produced by step S2a.
    Then, since relation $s_i \hb c_{k_i} <_{h'} s_{i+1}$ holds we may consider the two following cases:
    \begin{compactitem}
    \item Relation $s_i <_{h'} c_{k_i} <_{h'} s_{i+1}$ holds  with $r_i(x_j), w_{k_i}(x_{k_i}) \in h$ and $c_j <_h c_{k_i}$.
      From $c_{k_i} <_{h'} s_{i+1}$ and step S2a, there exists an object $y$ such that $r_{i+1}(y_{k_i})$.
      Thus, by definition of the snapshot precedence relation, $T_i \rightarrow T_{i+1}$ holds.
    \item Relation $s_i \hb c_{k_i}$ equals $s_i <_{h'} o_i <_h c_{k_i}$ where $o_i$ is some read operation of $T_i$,
      Since $c_{k_i} <_{h'} s_{i+1}$ is produced by step S2a, we know that for some object $y$, 
      $r_{i+1}(y_{k_i})$ belongs to $h$.
      According to the definition of the snapshot precedence, $T_i \rightarrow T_{i+1}$ holds.
    \end{compactitem}

    Applying the result above to \refequation{reconstruct:1}, we obtain:
    $T_1 \rightarrow T_2 \ldots \rightarrow T_n \rightarrow T_1$.
    History $h$ violates \MON, a contradiction.

  \item  $h'$ satisfies rules D1 and D2.
    
    \begin{compactitem}[$-$]
      
    \item[($h'$ satisfies D1.1)]
      Follows from $h \in \ACA$,

    \item[($h'$ satisfies D1.2)]
      Immediate from step S1.

    \item[($h'$ satisfies D1.3)]
      Consider three transactions $T_i$, $T_j$ and $T_k$ such that 
      operations $r_i(x_j)$, $w_j(x_j)$ and $w_k(x_k)$ are in $h$.
      The definition of a history tells us that either $x_k \versionOrder_h x_j$
      or the converse holds.
      We consider the following two cases:
      \begin{compactitem}
      \item[(Case $x_k \versionOrder_h x_j$)]
        Since $h$ belongs to $\ACA \inter \CONS \inter \WCF$, 
        \reflem{reconstruct:3} tells us that $c_k <_h c_j$ holds.
        Hence, $c_k <_{h'} c_j$ holds.
      \item[(Case  $x_j \versionOrder_h x_k$)]
        Applying again \reflem{reconstruct:3}, we obtain that  $c_j <_h c_k$ holds.
        Since $<_h$ is a partial order, then $c_j <_{h} c_k$ does not hold.
        By step S2b, the order $(s_i,c_k)$ is in $h'$.
      \end{compactitem}
      
    \item[($h'$ satisfies D2)]
      Consider two conflicting transaction $(T_i,T_j)$ in $h'$.
      Since $h$ belongs to  $\WCF$, 
      one of the following two cases occurs:
      \begin{compactitem}
      \item[(Case $T_i \depend T_j$)]
        At first glance, assume that $T_i \depend T_j$ holds.
        By step S2a, $s_i$ is in $h'$ after every operation $c_j$ such that $r_i(x_j)$ is in $h'$,
        and by step S1, $s_i$ precedes the first operation of $T_i$.
        Thus $c_j <_{h'} s_i$ holds, and $h'$ satisfies D2 in this case. 
        To obtain the general case, we applying inductively the previous reasoning.
      \item[(Case $T_j \depend T_i$)]
        The proof is symmetrical to the case above, and thus omitted.
      \end{compactitem}

    \end{compactitem}

  \end{itemize}

\end{proof}

\vspace{0.5em}
\noindent
From the conjunction of \refprop{reconstruct:1} and \refprop{reconstruct:2}, we deduce 
our decomposition theorem.

\begin{theorem}
  \labtheo{reconstruct:main}
  $\SI = \ACA \inter \SCONS \inter \MON \inter \WCF$
\end{theorem}

\noindent
Notice that this decomposition is well-formed in the sense that 
the four properties \SCONS, \MON, \WCF and \ACA are all distinct
and that no strict subset of $\{\SCONS, \MON, \WCF, \ACA \}$ attains \SI.

\begin{proposition}
  For every $S \subsetneq \{\SCONS, \MON, \WCF, \ACA \}$, 
  it is true that $\inter_{X \in S} X \neq \SI$.
\end{proposition}

\begin{proof}  
  For every set $S \subsetneq \{\SCONS, \MON, \WCF, \ACA \}$ containing three of the four properties, 
  we exhibit below a history in $\inter_{X \in S} X \setminus \SI$.
  Trivially, the result then holds for every $S$.
  \begin{itemize}
  \item[-] $\SCONS \inter \ACA \inter \WCF$: History $h_7$ in \refsection{reconstruct:scons}.
  \item[-] $\MON \inter \ACA \inter \WCF$: History $h_6$ in \refsection{reconstruct:scons}.
  \item[-] $\SCONS \inter \MON \inter \WCF$: History $r_1(x_0).w_1(x_1).r_a(x_0).c_1.c_a$.
  \item[-] $\SCONS \inter \MON \inter \ACA$: History $r_1(x_0).r_2(x_0).w_1(x_1).w_2(x_2).c_1.c_2$.
  \end{itemize}
\end{proof}

To the best of our knowledge, this result is the first to prove that \SI can be split into simpler properties.
\reftheo{reconstruct:main} also establishes that \SI is definable on plain histories.
This has two interesting consequences:
(i) a transactional system does not have to explicitly implement snapshots to support \SI,
and (ii) one can compare \SI to other consistency criterion without relying on a phenomena based characterization
(contrary to, e.g., the work of Adya \cite{Adya99}).

\newpage
\section{The impossibility of \SI with \GPR}
\labsection{imp}

This section leverages our previous decomposition result to show that \SI is inherently non-scalable.
In more details, we show that none of \MON, \SCONSa or \SCONSb
is attainable in some asynchronous failure-free \GPR system \procSet
when updates are obstruction-free and queries are wait-free.
To prove these results, we first characterize in Lemmata \ref{lem:imp:0} and \ref{lem:imp:1}
histories acceptable by \procSet.

\begin{lemma}[Positive-freshness Acceptance]
  \lablem{imp:0}
  Consider an acceptable history $h$ and a transaction $T_i$ pending in $h$
  such that the next operation invoked by $T_i$ is a read on some object $x$.
  Note $x_j$  the latest committed version of $x$ prior to the first operation of $T_i$ in $h$.
  Let $\run$ be an execution satisfying $\mathfrak{F}(\run)=h$.
  If $h.r_i(x_j)$ belongs to $\SI$ and there is no concurrent write-conflicting transaction with $T_i$,
  then 
  there exists an execution $\run'$ extending $\run$ such that in history $\refMapOf{\run'}$,
  transaction $T_i$ reads at least  (in the sense of $\versionOrder_h$) version $x_j$ of $x$.
\end{lemma}

\begin{proof}
  By contradiction.
  Assume that in every execution extending $\run$, transaction $T_i$ reads a version $x_k \versionOrder_{h} x_j$.
  Let $\run'$ be such an extension in which 
  \begin{inparaenum}[(i)]
  \item no other transaction than $T_i$ makes steps,
  \item we extend $T_i$ after its read upon $x$ by a write on $x$, then
  \item $\coordOf{T_i}$ tries committing $T_i$.
  \end{inparaenum}
  Since $T_i$ reads version $x_k$ in $\refMapOf{\run'}$, transaction $T_i$ should abort.
  However in history $\refMapOf{\run'}$ there is no concurrent write-conflicting transaction with $T_i$.
  Hence, this execution contradicts that updates are obstruction-free.
\end{proof}

\begin{lemma}[Genuine Acceptance]
  \lablem{imp:1}
  Let $h=\refMapOf{\run}$ be an acceptable history by \procSet such that a transaction $T_i$ is pending in $h$.
  Note $X$ the set of objects accessed by $T_i$ in $h$.
  Only processes in \replicaSetOf{X} make steps to execute $T_i$ in $\run$.
\end{lemma}

\begin{proof}
  (By contradiction.)
  Consider that a process $p \notin \replicaSetOf{X}$ makes steps to execute $T_i$ in $\run$.
  Since the prefix of a transaction is a transaction with the same id, 
  we can consider an extension $\run'$ of $\run$ such that
  $T_i$ does not execute any additional operation in $\run'$ and $\dbCoordinatorOf{T_i}$ is correct in $\run'$.
  The progress requirements satisfied by \procSet imply that $T_i$ terminates in $\run'$.
  However, process $p \notin \replicaSetOf{X}$  makes steps to execute $T_i$ in $\run'$.
  A contradiction to the fact that \procSet is GPR.
\end{proof}

We now state that monotonic snapshots are not constructable by \procSet.
Our proof holds because objects accessed by a transaction are not known in advance.
\begin{theorem}
  \labtheo{imp:1}
  No asynchronous failure-free \GPR system implements \MON
\end{theorem}

\begin{proof}
  (By contradiction.)
  Let us consider
  (i) four objects $x$, $y$, $z$ and $u$ such that for any two objects in $\{x,y,z,u\}$,
  their replica sets do not intersect;
  (ii) four queries $T_a$, $T_b$, $T_c$ and $T_d$
  accessing respectively $\{x,y\}$, $\{y,z \}$, $\{z,u\}$ and $\{u,x\}$;
  and (iii) four updates $T_1$, $T_2$, $T_3$ and $T_4$ modifying respectively $x$, $y$, $z$ and $u$.

  Obviously, history $r_b(y_0)$ is acceptable,
  and since updates are obstruction-free, $r_b(y_0).r_2(y_0).w_2(y_2).c_2$ is also acceptable.
  Applying that \reflem{imp:0},
  we obtain that history $r_b(y_0).r_2(y_0).w_2(y_2).c_2.r_a(x_0).r_a(y_2)$ is acceptable.
  Since $T_a$ is wait-free, 
  $h=r_b(y_0).r_2(y_0).w_2(y_2).c_2.r_a(x_0).r_a(y_2).c_a$ is acceptable as well.
  Using a similar reasoning, 
  $h'=r_d(u_0).r_4(u_0).w_4(u_4).c_4.r_c(z_0).r_c(u_4).c_c$ is also acceptable.
  We note $\run$ and $\run'$ respectively two sequences of steps 
  such that $\refMapOf{\run}=h$ and $\refMapOf{\run'}=h'$.

  The system \procSet is \GPR. 
  As a consequence, \reflem{imp:1} tells us that only processes in
  $\replicaSetOf{x,y}$ make steps in $\run$.
  Similarly, only processes in $\replicaSetOf{u,z}$ make steps in $\run'$.
  By hypothesis, $\replicaSetOf{x,y}$ and $\replicaSetOf{u,z}$ are disjoint.
  Applying a classical indistinguishably argument \cite[Lemma~1]{Fischer1985},
  both $\run'.\run$ and $\run.\run'$ are admissible by \procSet.
  Thus, histories $h'.h=\refMapOf{\run'.\run}$ and $h.h'=\refMapOf{\run.\run'}$ are acceptable.

  Since updates are obstruction-free, history $h'.h.r_3(z_0).w_3(z_3).c_3$ is acceptable.
  Note $U$ the sequence of steps following $\run'.\run$ with $\refMapOf{U}=r_3(z_0).w_3(z_3).c_3$.
  Observe that by \reflem{imp:1} $\run'.\run.U$ is indistinguishable from $\run'.U.\run$.
  Then consider history $\refMapOf{\run'.U.\run}$.
  In this history, $T_b$ is pending and the latest version of object $z$ is $z_3$, 
  As a consequence, by applying \reflem{imp:0},
  there exists an extension of $\run'.U.\run$ in which transaction $T_b$ reads $z_3$.
  From the fact that queries are  wait-free and since $\run'.\run.U$ is indistinguishable from $\run'.U.\run$,
  we obtain that history $h_1=h'.h.r_3(z_0).w_3(z_3).c_3.r_b(z_3).c_b$ is acceptable.
  We note $U_1$ the sequence of steps following $\run'.\run$ such that 
  $\refMapOf{U_1}$ equals $r_3(z_0).w_3(z_3).c_3.r_b(z_3).c_b$.
 
  With a similar reasoning, history $h_2=h'.h.r_1(x_0).w_1(x_1).c_1.r_d(x_1).c_d$ is acceptable.
  Note $U_2$ the sequence satisfying $\refMapOf{U_2}=r_1(x_0).w_1(x_1).c_1.r_d(x_1).c_d$.

  Executions $\run'.\run.U_1$ and $\run'.\run.U_2$ are both admissible.
  Because \procSet is \GPR,
  only processes in $\replicaSetOf{y,z}$ (resp. $\replicaSetOf{x,u}$) make steps in $U_1$ (resp. $U_2$).
  By hypothesis, these two replica sets are disjoint.
  Applying again an indistinguishably argument, $\run'.\run.U_1.U_2$ is an execution of \procSet.
  Therefore, 
  the history $\hat{h}=\refMapOf{\run'.\run.U_1.U_2}$ is acceptable.
  In this history, relation $T_a \rightarrow T_b \rightarrow T_c \rightarrow T_d \rightarrow T_a$ holds.
  Thus, $\hat{h}$ does not belong to \MON. Contradiction. 
\end{proof}

Our next theorem states that \SCONSb is not attainable.
Similarly to Attiya et al. \cite{Attiya:SPAA2009}, our proof builds an infinite
execution in which a query $T_a$ on two objects never terminates.
We first define a finite execution during which we interleave between any two consecutive steps to execute $T_a$, 
a transaction updating one of the objects read by $T_a$.
We show that during such an execution, transaction $T_a$ does not terminate successfully.
Then, we prove that asynchrony allows us to continuously extend such an execution,
contradicting the fact that queries are wait-free.

\begin{definition}[Flippable execution]
  Consider
  two distinct objects $x$ and $y$,
  a query $T_a$ over both objects,
  and a set of updates $T_{j \in \llbracket 1,m \rrbracket}$ accessing $x$ if $j$ is odd, and $y$ otherwise.
  An execution $\run=U_1V_2U_2 \ldots V_mU_m$ where,
  \begin{itemize}
    \item transaction $T_a$ reads in history $h=\mathfrak{F}(\run)$ at least version $x_1$ of $x$,
    \item for any $j$ in $\llbracket 1,m \rrbracket$,  $U_j$ is the execution of transaction $T_j$ by processes $Q_j$, 
    \item for any $j$ in $\llbracket 2,m \rrbracket$, $V_j$ are steps to execute $T_a$ by processes $P_j$, and
    \item both $(Q_j \inter P_j = \emptySet) \xor (P_j \inter Q_{j+1} = \emptySet)$ and $Q_j \inter Q_{j+1} = \emptySet$ hold,
  \end{itemize}
  is called flippable.
\end{definition}

\begin{lemma}
  \lablem{imp:2}
  Let $\run$ be an execution admissible by \procSet.
  If $\run$ is flippable and histories accepted by \procSet satisfy \SCONSb, 
  query $T_a$ does not terminate.
\end{lemma}

\begin{proof}
  Let $h$ be the history $\mathfrak{F}(\run)$.
  In history $h$ transaction $T_j$ precedes transaction $T_{j+1}$,
  it follows that $h$ is of the form $h=w_1(x_1).c_1. *. w_2(y_2).c_2.* \ldots$\ ,
  where each symbol $*$ corresponds to either no operation, or to some read operation 
  by $T_a$ on object $x$ or $y$.

  Because \run is flippable, transaction $T_a$ reads at least version $x_1$ of object $x$ in $h$.
  For some odd natural  $j \geq 1$,  let $x_j$ denote the version of object $x$ read by $T_a$.
  Similarly, for some even natural $l$, let $y_{l}$ be the version of $y$ read by $T_a$.
  Assume that $j<l$ holds.
  Therefore, $h$ is of the form $h=\ldots  w_{j}(x_{j}) \ldots w_l(y_l) \ldots$.

  Note $k$ the value $l+1$, 
  and consider the sequence of steps $V_k$ made by $P_k$ right after $U_{l}$ to execute $T_a$.
  Applying the definition of a flippable execution, we know that  
  (F1) $(Q_l \inter P_k = \emptySet) \xor (P_k \inter Q_k = \emptySet)$,
  and (F2) $Q_l \inter Q_k \equals \emptySet$.
  Consider now the following cases:
  \begin{compactitem}
  \item[(Case $Q_l \inter P_k = \emptySet$.)]
    It follows that $\run$ is indistinguishable from the execution 
    $\run''=\ldots U_j \ldots V_{k} U_l U_k \ldots$. 
    Then from fact F2, 
    $\run$ is indistinguishable from execution $\run'=\ldots U_j \ldots V_{k} U_k U_l \ldots$.
  \item[(Case $P_k \inter Q_k = \emptySet$)]
    With a similar reasoning, we obtain that $\run$ is indistinguishable from
    $\run'=\ldots U_j \ldots U_k U_l V_k \ldots$.
  \item[(Case $P_k \inter (Q_l \union Q_k) = \emptySet$.)]
    This case reduces to any of the two above cases.
  \end{compactitem}
  Note $h'$ the history $\mathfrak{F}(\run')$.
  Observe that since $\run'$ is indistinguishable from $\run$, history $h'$ is acceptable.
  In history $h'$, $c_k <_{h'} c_l$ holds.
  Moreover, $c_j <_{h'} c_k$ holds by the assumption $j<l$ and the fact that $k$ equals $l+1$.
  Besides, operations  $r_i(x_j)$, $r_i(y_l)$ and $w_k(x_k)$ all belong to $h'$.
  According to the definition of \SCONSb, transaction $T_a$ does not commit in $h'$.
  (The case $j>l$ follows a symmetrical reasoning to the case $l>j$ we considered previously.)
\end{proof}

\begin{theorem}
  \labtheo{imp:2}
  No asynchronous failure-free \GPR system implements \SCONSb.
\end{theorem}

\begin{proof}
  (By contradiction.)
  Consider two objects $x$ and $y$ such that
  $\replicaSetOf{x}$ and $\replicaSetOf{y}$  are disjoint.
  Assume a read-only transaction $T_a$ that reads successively $x$ then $y$.
  Below, we exhibit an execution admissible by \procSet during which transaction $T_a$ never terminates.
  We build this execution as follows:

  [Construction.]
  Consider some empty execution $\run$.
  Repeat for all $i>=1$:
  Let $T_i$ be an update of $x$, if $i$ is odd, and $y$ otherwise.
  Start the execution of transaction $T_i$.
  Since no concurrent transaction is write-conflicting with $T_i$ in $\run$ and updates are obstruction-free,
  there must exist an extension $\run.U_i$ of $\run$ during which $T_i$ commits.
  Assign to $\run$ the value of $\run.U_i$.
  Execution $\run$ is flippable.
  Hence, \reflem{imp:2} tells us that transaction $T_a$ does not terminate in this execution.
  Consider the two following cases:
  (Case $i=1$)
  Because \procSet satisfies non-trivial \SI,
  there exists an extension $\run'$ of $\run$
  in which transaction $T_a$ reads at least version $x_1$ of object $x$.
  Notice that execution $\run'$ is of the form $U_1.V_2.s.\ldots$ where
  \begin{inparaenum}[\em (i)]
  \item all steps in $V_2$ are made by processes in $\replicaSetOf{x}$, and
  \item $s$ is the first step such that $\refMapOf{U_1.V_2.s.}=r_1(x_0).w_1(x_1).c_1.r_a(x_1)$.
  \end{inparaenum}
  Assign $U_1.V_2$ to $\run$ .
  (Case $i>2$)
  Consider any step $V_{i+1}$ to terminate $T_a$ and append it to $\run$.

  Execution $\run$ is admissible by \procSet.
  Hence $\refMapOf{\run}$ is acceptable.
  However, in this history transaction $T_a$ does not terminate.
  This contradicts the fact that queries are wait-free.
\end{proof}

\SCONSa disallows some real time orderings between operations accessing different objects.
Our last theorem shows that this property cannot be maintained under \GPR.
\begin{theorem}
  \labtheo{imp:3}
  No asynchronous failure-free \GPR system implements \SCONSa.
\end{theorem}

\begin{proof}
  (By contradiction.)
  Consider two distinct objects $x$ and $y$ such that $\replicaSetOf{x}$ and $\replicaSetOf{y}$ are disjoint.
  Let $T_1$ be an update accessing $y$, and $T_a$ be a query reading both objects.
  
  Obviously, history $h=r_a(x_0)$ is acceptable.
  Note $U_a$ a sequence of steps satisfying $U_a=\refMapOf{r_a(x_0)}$.
  Because \procSet supports obstruction-free updates,
  we know the existence of an extension $U_a.U_1$ of $U_a$ such that $\refMapOf{U_1}=r_1(y_0).w_1(y_1).c_1$.
  By \reflem{imp:1}, we observe that $U_a.U_1$ is indistinguishable from $U_1.U_a$.
  Then by \reflem{imp:0}, there must exist an extension $U_1.U_a.V_a$ of $U_1.U_a$ admissible by \procSet
  and such that $\refMapOf{V_a}=r_a(y_1).c_a$.
  Finally, since $U_a.U_1$ is indistinguishable from $U_1.U_a$ and $U_1.U_a.V_a$ is admissible,
  $U_a.U_1.V_a$ is admissible too.
  The history $\refMapOf{U_a.U_1.V_a}$ is not in \SCONSa.
  Contradiction.
\end{proof}

As a consequence of the above, no asynchronous system, even if it is failure-free, can support both \GPR and \SI.
In particular, even if the system is augmented with failure detectors \cite{CT96}, 
a common approach to model partial synchrony, \SI cannot be implemented under \GPR.
This fact strongly hinders the usage of \SI at large scale.
In the following sections, 
we further discuss implications of this impossibility result
then we introduce a novel consistency criterion to overcome it.

\newpage 
\section{Discussion}
\labsection{corollaries}

In this section, we discuss the consequences of our impossibility results, 
with an emphasis on other consistency criteria than \SI.

\subsection{Declaring the Read-set in Advance}
\labsection{corollaries:ser:declarin}

When a transaction declares objects it accesses \emph{in advance}, 
a \GPR system can install a strictly consistent and monotonic snapshot 
just after the start of the transaction.
As a consequence, such an assumption sidesteps our impossibility result.
This is the approach employed in the \SI protocol of \citet{Armendariz-Inigo2008}.
Still, this protocol makes use of atomic broadcast to install a snapshot.
We obtain a \GPR system that supports \SI by replacing this group communication primitive by a genuine atomic multicast.

\subsection{Strict Serializability and Opacity}
\labsection{corollaries:sser}

We observe that \reftheo{imp:3} also holds if we consider the following (classical) definition of obstruction-free updates
in which both read-write and write-write conflicts are taken into account:

\begin{itemize}
\item \textbf{Obstruction-free Updates (\OFU-a).}
  For every update transaction $T_i$,  
  if \coordOf{T_i} is correct then $T_i$ eventually terminates.
  Moreover, if $T_i$ does not \emph{conflict} with some concurrent transaction then $T_i$ eventually commits.
\end{itemize}

As a consequence,  neither strict serializability \cite{Papadimitriou1979}, nor opacity \cite{opacity} is attainable under \GPR.
In the case of opacity, this answers negatively to a problem recently posed by Peluso et al. \cite{RomanoWTTM12}.

\subsection{Serializability (\SER)}
\labsection{corollaries:ser}

\subsubsection{Permissiveness}
\labsection{corollaries:ser:permissiveness}

A transactional system \procSet is \emph{permissive} with respect to a consistency criterion $\mathcal{C}$ 
when every history $h \in \mathcal{C}$ is acceptable by \procSet.
Permissiveness \cite{GuerraouiHS08} measures the optimal amount of concurrency a system allows.
If we consider again histories $h_1$ and $h_2$ in the proof of \reftheo{imp:1},
we observe that both histories are serializable.
Hence, every system  permissive with respect to \SER accepts both histories.
By relying on the very same argument as the one we exhibit to close the proof of \reftheo{imp:1},
we conclude that no transactional system is both \GPR and permissive with respect to \SER.
For instance, P-Store \cite{Schiper2010}, a \GPR protocol that ensures \SER, 
does not accept history $h_{10}=r_1(x_0).w_1(x_1).c_1.r_2(x_0).r_2(y_0).w_2(y_2).c_2$.

\subsubsection{Wait-free Queries.}
\labsection{corollaries:ser:wfq}

Under \SI, a query never forces an update to abort.
This key feature of \SI greatly improves performance.
Most recent transactional systems that support \SER 
(e.g., \cite{AAAS97,SAA98,KA00,FI01,DBSM03,MJKA05,LKPJ05,LPM07,PF08, Schiper2010, PedoneDSN12, Peluso2012a})
offer such a progress property as well as positive-freshness acceptance:%
\footnote{
  \reflem{imp:0} proves positive-freshness acceptance for \SI under standard assumptions (\OFU and \WFQ).
  In the case of \SER, this property is a feature of the input acceptance of the protocol.
}

\begin{itemize}
\item \textbf{Obstruction-free Updates (OFU-b).}
  For every update transaction $T_i$,  
  if \coordOf{T_i} is correct then $T_i$ eventually terminates.
  Moreover, if $T_i$ does not conflict with some concurrent \emph{update} transaction then $T_i$ eventually commits.
\item \textbf{Positive Freshness Acceptance.} 
  Consider an acceptable history $h$ and a transaction $T_i$ pending in $h$
  such that the next operation invoked by $T_i$ is a read on some object $x$.
  Note $x_j$  the latest committed version of $x$ prior to the first operation of $T_i$ in $h$.
  Let $\run$ be an execution satisfying $\mathfrak{F}(\run)=h$.
  If $h.r_i(x_j)$ belongs to $\SER$ and there is no concurrent write-conflicting update transaction with $T_i$,
  then 
  there exists an execution $\run'$ extending $\run$ such that in history $\refMapOf{\run'}$,
  transaction $T_i$ reads at least  (in the sense of $\versionOrder_h$) version $x_j$ of $x$.
\end{itemize}

When the two above progress properties holds, \reftheo{imp:1} applies to \SER transactional systems,
implying a choice between \WFQ and \GPR.
The P-Store transactional system of \citet{Schiper2010} favors \GPR over \WFQ.
On the contrary, the protocol of \citet{PedoneDSN12} ensures \WFQ but is not \GPR.
Recently, \citet{Peluso2012a} have proposed a \GPR algorithm that supports both \SER and \WFQ in the failure-free case.
This protocol sidesteps the impossibility result by dropping obstruction-freedom
for updates in certain scenarios.%
\footnote{
  In more details, this algorithm numbers every version with a scalar.
  If a transaction $T_i$ first reads an object $x$ then updates an object $y$, 
  in case the version of $x$ is smaller than the latest version of $y$, say $y_k$,
  $T_i$ will not be able to read $y_k$ , and it will thus abort.
}

\subsection{Parallel Snapshot Isolation (PSI)}
\labsection{corollaries:psi}

Recently, \citet{Sovran2011} have introduced a weaker consistency criterion than \SI named parallel snapshot isolation (\PSI).
\PSI  allow snapshots to be non-monotonic, but still require them to ensure \SCONSa.
Sovran et al. justify the use of \PSI in Walter
by the fact that \SI is too expensive in a geographically distributed environment \cite[page 4]{Sovran2011}.
Our impossibility result establishes that, in order to scale,
a transactional system needs supporting both non-monotonic \emph{and} non-strictly consistent snapshots.
Thus, while being more scalable than \SI, \PSI yet cannot be implemented in a \GPR system.

\newpage
\section{\NMSILongName}
\labsection{wsi}

We just showed that the \SI requirements of strictly consistent (\SCONS)
and monotonic (\MON) snapshots hurt scalability, as they are impossible
with \GPR.
To overcome the impossibility, this section presents a slightly weaker criterion, called \NMSILongName~(\NMSI).

\NMSI retains the most important properties of \SI, namely snapshots
are consistent, a read-only transaction can commit locally without
synchronization, and two concurrent conflicting updates do not both commit.
However, \NMSI allows non-strict, non-monotonic snapshots.
For instance, history $h_7$ in \refsection{reconstruct:mon}, which is
not in \SI, is allowed by \NMSI.
Formally, we define \NMSI as follows:

\begin{definition*}[\NMSILongName]
  A history $h$ is in $\NMSI$ iff $h$ belongs to $\ACA \inter \CONS \inter \WCF$.
\end{definition*}

To clarify our understanding of \NMSI, \reftab{anomalies} compares it to
well-known approaches, based on the anomalies an application might observe.
In addition to the classical anomalies \cite{Berenson1995,Adya99} (dirty
reads, non-repeatable reads, read skew, dirty writes, lost updates, and
write skew), we also consider the following:
(Non-Monotonic Snapshots) snapshots taken by transactions are not monotonically ordered,
and
(Real-Time Causality Violation) a transaction $T_2$ observes the effect of
some transaction $T_1$, but does not observe the effect of all the transactions that
precede (in real time) $T_1$.

\vspace{0.5em}
\begin{table*}[h!]
  \footnotesize
  \centering
  \begin{tabular}{ @{}c|ccccc}
    & \multicolumn{1}{p{16mm}}{\centering Strict Serializablity \cite{Papadimitriou1979}} 
    & \multicolumn{1}{p{16mm}}{\centering Serializablity \cite{Berenson1995}}
    & \multicolumn{1}{p{16mm}}{\centering Update Serializablity \cite{GM80}}
    & \multicolumn{1}{p{13mm}}{\centering Snapshot Isolation}  
    & \multicolumn{1}{p{10mm}}{\centering \NMSI} \\
    \hline
    \multicolumn{1}{c|}{Dirty Reads} & x & x & x & x & x \\
    \multicolumn{1}{c|}{Non-repeatable Reads} & x & x & x & x & x\\
    \multicolumn{1}{c|}{Read Skew} & x & x & x & x & x\\
    \hline
    \multicolumn{1}{c|}{Dirty Writes} & x & x & x & x & x  \\
    \multicolumn{1}{c|}{Lost Updates} & x & x & x & x & x  \\
    \multicolumn{1}{c|}{Write Skew} & x & x & x & - & -  \\
    \hline
    \multicolumn{1}{c|}{Non-Monotonic Snapshots} & x & x & - & x & -  \\
    \multicolumn{1}{c|}{Real-time Causality Violation} & x & - & - & x & -  \\    
  \end{tabular}
  \caption{Comparing consistency criterion by their anomalies  (x: disallowed)}
  \labtab{anomalies}
\end{table*}
\vspace{0.5em}

Write Skew, the classical anomaly of \SI, is observable under \NMSI.
(\citet{Cahill2008} show how an application can easily avoid it.)
Because \NMSI does not ensure \SCONSb, it suffers
the Real-Time Causality Violation anomaly.
Note that it is not new, as it occurs with serializability as well;
this argues that it is not considered a problem in practice.
Non-Monotonic Snapshots occur both under \NMSI and update serializability.
Following \citet{GM80}, we believe that this is a small price to pay for improved performance.

\newpage
\section{Protocol}
\labsection{protocol}

We now describe \jessy, a scalable transactional system that implements \NMSI with \GPR.
Because distributed locking policies do not scale \cite{233330,1048870},
\jessy employs deferred update replication:
transactions are executed optimistically, then certified by a termination protocol.
\jessy uses a novel clock mechanism to ensure that snapshots are both fresh and consistent,
while preserving wait-freedom of queries and genuineness.
We describe it in the next section.

\subsection{Building Consistent Snapshots} 
\labsection{protocol:snap}

To compute consistent snapshots, \jessy makes use of a novel data type called \emph{dependence vectors}.
Each version of each object is assigned its own dependence vector.
The dependence vector of some version $x_i$ reflects all the versions read by $T_i$, or
read by transactions that precede $T_i$, as well as the writes of $T_i$ itself:

\begin{definition*}[Dependence Vector]
  A dependence vector is a function \dv that maps every read (or write) operation $o(x)$ in a history $h$
  to a vector $\dvOf{o(x)} \in \naturalSet^{\cardinalOf{\objectSet}}$ such
  that:
  \begin{displaymath}
    \hspace{-0.5em}%
      \begin{array}{l}
        \dvOf{r_i(x_0)} = 0^{\cardinalOf{\objectSet}} \\
        \dvOf{r_i(x_j)} = \dvOf{w_j(x_j)} \\
        \dvOf{w_i(x_i)} = \mathit{max}~\{ \dvOf{r_i(y_j)} : y_j \in \readSetOf{T_i} \}  \\
        ~~~~~~~~~~~~~~~~~~~ ~+~ \Sigma_{z_i \in \writeSetOf{T_i}} ~1_z
      \end{array}
  \end{displaymath}
  where
  $\mathit{max}~\mathcal{V}$ is the vector containing for each dimension $z$,
  the maximal $z$ component in the set of vectors $\mathcal{V}$,
  and $1_z$ is the vector that equals $1$ on dimension $z$ and $0$ elsewhere.
\end{definition*}

To illustrate this definition, consider history $h_{10}$ below.
In this history, transactions $T_1$ and $T_2$ update objects $x$ and $y$
respectively, while transaction $T_3$ reads $x$, then updates $y$.
\begin{figure}[!ht]
  \vspace{-1em}
  \centering
  \begin{tikzpicture}
    \draw (0,2) node[right]{\small $h_{10}=$};

    \draw (1,2) node[right]{\small $r_1(x_0).w_1(x_1).c_1$};
    \draw (1,1) node[right]{\small $r_2(y_0).w_2(y_2).c_2$};

    \draw [->] (4,2) -- (4.45,1.6);
    \draw [->] (4,1) -- (4.45,1.4);

    \draw (4.5,1.5) node[right]{\small $r_3(x_1).r_3(y_2).w_3(y_3).c_3$};

  \end{tikzpicture}
  \vspace{-1em}
\end{figure}

\noindent%
The dependence vector of $x_1$ equals $\langle 1,0 \rangle$,
and of $y_1$ equals $\langle 0, 1 \rangle$.
Since transaction $T_3$ reads $x$ then updates $y$, 
this implies that dependence vector of $y_3$ equals
$\langle 1,2 \rangle$.

\begin{definition*}[Compatibility Relation]
Consider a transaction $T_i$ and two versions $x_j$ and $y_l$ read by $T_i$.
We shall say that $x_j$ and $y_l$ are \emph{compatible} for  $T_i$, written $\isCompatible{T_i,x_j,y_l}$,
when both $\dvOf{r_i(x_j)}[x] \geq \dvOf{r_i(y_l)}[x]$ and $\dvOf{r_i(y_l)}[y] \geq \dvOf{r_i(x_j)}[y]$ hold.
\end{definition*}

Using the compatibility relation, we can prove that dependence vectors fully characterize consistent snapshots.
First of all, we show in \reflem{dv:1} that if transaction $T_i$ depends on transaction $T_j$ 
then the dependence vector of any object written by $T_i$ is greater than the dependence vector of any object written by $T_j$. 

\begin{lemma}
  \lablem{dv:1}
  Consider
  a history $h$ in \NMSI, 
  and two transactions $T_i$ and $T_j$ in $h$.
  Then,
  \begin{equation*}
    T_i \depend T_j
    \iff
    \forall x,y \in Objects :
    \forall w(x),w(y) \in T_i \times T_j:
    \dvOf{w_i(x_i)} > \dvOf{w_j(y_j)}    
  \end{equation*}
\end{lemma}

\begin{IEEEproof}
  The proof goes as follows:
  \begin{compactitem}
  \item ($\Rightarrow$)
    First consider that $T_i \MustHave T_j$ holds.
    By definition of relation $\MustHave$, we know that for some object $z$, 
    operations $r_i(z_j)$ and $w_j(z_j)$ are in $h$. 
    According to definition of function \dv we have: $\dvOf{w_i(x_i)} \geq \dvOf{r_i(z_j)} + 1_x$.
    Besides, always according to the definition of \dv,
    it is true that the following equalities hold: $\dvOf{r_i(z_j)}=\dvOf{w_j(z_j)}= \dvOf{w_j(y_j)}$.
    Thus, we have: $\dvOf{w_i(x_i)} > \dvOf{w_j(y_j)}$.
    The general case  $T_i \depend T_j$ is obtained by applying inductively the previous reasoning.
  \item ($\Leftarrow$)
    From the definition of function \dv, it must be the case that $r_i(y'_{j'})$ is in $h$ with $j' \neq 0$.
    We then consider the following two cases:
    (Case $j' = j$) By definition of relation $\MustHave$, $T_i  \MustHave T_j$ holds.
    (Case $j' \neq j$)
    By construction, we have that: $T_i \MustHave T_{j'}$.
    By definition of function \dv, we have that
    $\dvOf{r_{j'}(y_{j'})} = \dvOf{w_{j'}(y_{j'})}$.
    Since $\dvOf{w_i(x_i)} > \dvOf{w_j(y_j)}$ holds, 
    $\dvOf{w_{j'}(y_{j'})}[y] \geq   \dvOf{w_{j}(x_{j})}[y]$ is true.
    Both transactions $T_j$ and $T_{j'}$ write $y$.
    Since $h$ belongs to \NMSI, 
    it must be the case that either $T_j \depend T_{j'}$ or that $T_{j'} \depend T_j$ holds.
    If $T_j \depend T_{j'}$ holds, then we just proved that 
    $\dvOf{w_{j}(y_{j})} > \dvOf{w_{j'}(y_{j'})}$ is true.
    A contradiction.
    Hence necessarily $T_{j'} \depend T_j$ holds.
    From which we conclude that $T_i \depend T_j$ is true.
  \end{compactitem}
\end{IEEEproof}

\noindent
The following theorem shows that dependence vectors enable taking consistent snapshots.

\begin{theorem}
  \labtheo{protocol:1} 
  Consider a  history $h$ in \NMSI and a transaction $T_i$ in $h$.
  Transaction $T_i$ sees a consistent snapshot in $h$
  \emph{if, an only if,}
  every pair of versions $x_j$ and $y_l$ read by $T_i$ is compatible.
\end{theorem}

\begin{IEEEproof}
  The proof goes as follows:
  \begin{compactitem}
  \item ($\Rightarrow$)
    By contradiction.
    Assume the existence of two versions $x_l$ and $y_{j}$ in the snapshot of $T_i$
    such that $\dvOf{r_i(x_l)}[x] < \dvOf{r_i(y_{j})}[x]$ holds.
    By definition of function \dv, we have $\dvOf{r_i(x_{l})} = \dvOf{w_{l}(x_{l})}$
    and $\dvOf{r_i(y_{j})} = \dvOf{w_{j}(y_{j})}$.
    Hence, $\dvOf{w_l(x_l)}[x] < \dvOf{w_{j}(y_{j})}[x]$ holds.
    Again from the definition of function \dv, 
    there exists a transaction $T_{k \neq 0}$ writing on $x$ such that 
    (i) $\dvOf{w_{j}(y_{j})} \geq \dvOf{w_{k}(x_{k})}$ 
    and (ii) $\dvOf{w_{j}(y_{j})}[x] = \dvOf{w_{k}(x_{k})}[x]$.
    Applying \reflem{dv:1} to (i), we obtain $T_j \depend T_k$.
    From which we deduce that $T_i \depend T_k$.
    Now since both transactions $T_l$ and $T_k$ write $x$ and $h$ belongs to \NMSI,
    $T_l \depend T_k$ or $T_k \depend T_l$ holds.
    From (ii) and $\dvOf{w_l(x_l)}[x] < \dvOf{w_{j}(y_{j})}[x]$,
    we deduce that $\dvOf{w_{l}(x_{l})}[x] < \dvOf{w_k(x_k)}[x]$.
    As a consequence of \reflem{dv:1}, $T_k \depend T_l$ holds.
    Hence $x_l \versionOrder_h x_k$.
    But $T_i \depend T_k$ and $r_i(x_l)$ is in $h$.
    It follows that $T_i$ does not read a consistent snapshot.
    Contradiction.

  \item ($\Leftarrow$)
    By contradiction.
    Assume that there exists
    an object $x$ and a transaction $T_k$ on which $T_i$ depends
    such that 
    $T_i$ reads version $x_j$, $T_k$ writes version $x_k$, and $x_j \versionOrder_h x_k$.    
    First of all, since $h$ is in \NMSI, one can easily show that $T_k \depend T_j$.
    Since $T_k \depend T_j$, 
    \reflem{dv:1} tells us that $\dvOf{w_{k}(x_{k})} > \dvOf{w_j(x_j)}$ holds.
    Since $T_i \depend T_k$ holds, a short induction on the definition of function \dv tells us that
    $\dvOf{r_{i}(x_{j})}[x] \geq \dvOf{w_{k}(x_{k})}|x]$ is true.
    From which we obtain that: $\dvOf{r_{i}(x_{j})}[x] \geq \dvOf{w_{k}(x_{k})}[x] > \dvOf{w_j(x_j)}[x] = \dvOf{r_{i}(x_{j})}[x]$.
    Contradiction.

  \end{compactitem}

\end{IEEEproof}


Despite that in the common case dependence vectors are sparse, they might be large for certain workloads.
For instance, if transactions execute random accesses,
the size of each vector tends asymptotically to the number of objects in the system.
To address the above problem, \jessy employs a mechanism to approximate dependencies safely, 
by coarsening the granularity, grouping objects into disjoint partitions 
and serializing updates in a group as if it was a single larger object.
We cover this mechanism in what follows.
\subsection{Partitioned Dependence Vector}

Consider some partition $\mathcal{P}$ of \objectSet.
For some object $x$, note  $\mathit{P}(x)$ the partition $x$ belongs to,
and by extension, for some $S \subseteq \objectSet$, note $\mathit{P}(S)$ the set $\{ \mathcal{P}(x) ~|~ x \in S\}$.
A partition is \emph{proper} when updates inside the same partition are serialized, that is,
for every $X \in \mathcal{P}$ and every two writes $w_i(x_i)$, $w_j(y_j)$ with $\mathcal{P}(x)=\mathcal{P}(y)$,
either $w_i(x_i) <_{h} w_j(y_j)$ or $w_j(y_j) <_h w_i(x_i)$ holds.

Now, consider some history $h$, and for every object $x$ replace every operation 
$o_i(x)$ in $h$ by $o_i(\mathcal{P}(x))$.
We obtain a history that we note $h^{\mathcal{P}}$.
The following result linked the consistency of $h$ to the consistency of $h^{\mathcal{P}}$:

\begin{proposition}
  \labprop{protocol:1}
  Consider some history $h$.
  If
  $\mathcal{P}$ is a proper partition of \objectSet for $h$
  and history $h^{\mathcal{P}}$ belongs to \CONS,
  then  $h$ is in \CONS.
\end{proposition}
\begin{proof}

  First of all we observe that for any two transactions $T_i$ and $T_j$:
  \begin{itemize}
  \item[-] If $T_i \depend T_j$ holds in $h$ then $T_i \depend T_j$ holds in $h^{\mathcal{P}}$.\\
    \emph{Proof.}
    If $T_i \MustHave T_j$ holds in $h$, then $r_i(x_j)$ is in $h$.
    Thus $r_i(\mathcal{P}(x_j))$ is in $h^{\mathcal{P}}$.
    It follows that $T_i \MustHave T_j$ holds in  $h^{\mathcal{P}}$.
    If $T_i \depend T_j$ in $h$
    then there exist a set of transactions 
    $\{T_1, \ldots, T_m\}$ such that:
    $T_i \MustHave T_1 \ldots \MustHave T_m \MustHave T_j$ hold in $h$.
    From the result above, 
    we deduce that $T_i \MustHave T_1 \ldots \MustHave T_m \MustHave T_j$ hold in  $h^{\mathcal{P}}$.
    Hence, $T_i \depend T_j$ holds in $h^{\mathcal{P}}$.
    \qed
  \item[-] If $x_i \versionOrder x_j$ holds in $h$ then $\mathcal{P}(x_i) \versionOrder \mathcal{P}(x_j)$
    holds in $h^{\mathcal{P}(x)}$. \\
    \emph{Proof.}
    If $x_i \versionOrder x_j$ holds in $h$ 
    then 
    $\mathcal{P}(x_i) \versionOrder \mathcal{P}(x_j)$ holds in $h$.
    \qed
  \end{itemize}
  
  For the sake of contradiction, 
  assume that  $h^{\mathcal{P}}$ is in \CONS while $h$ is not in \CONS.
  It follows that there exist
  a transaction $T_i$, 
  some  object $x$ and a transaction $T_k$ on which $T_i$ depends
  such that in $h$,
  $T_i$ reads version $x_j$, $T_k$ writes version $x_k$, and  $x_j \versionOrder_h x_k$.
  From the two observations above, we obtain that
  $T_i \MustHave T_j$, $T_i \depend T_k$ and $\mathcal{P}(x_j) \versionOrder_h \mathcal{P}(x_k)$
  hold in $h^{\mathcal{P}}$.
  Hence, $h^{\mathcal{P}}$ is not consistent.
  Contradiction.

\end{proof}

Given two operations $o_i(x_j)$ and $o_k(y_l)$, 
let us introduce relation $o_i(x_j) \leq_h^{\mathcal{P}} o_k(y_l)$
when $o_i(x_j)=o_k(y_l)$, or $o_i(x_j) <_h o_k(y_l) \land \mathcal{P}(x)=\mathcal{P}(y)$ is true.
Based on \refprop{protocol:1}, we define below a function that approximates dependencies safely:

\begin{definition}[Partitioned Dependence Vector]
  A partitioned dependence vector is a function \pdv that
  maps every read (or write) operation $o(x)$ in a history $h$ 
  to a vector $\pdvOf{o(x)} \in \naturalSet^{\cardinalOf{\mathcal{P}}}$ such that:
  \begin{displaymath}
    \begin{array}{l}
      \pdvOf{r_i(x_0)} = 0^{\cardinalOf{\mathcal{P}}} \\
      \pdvOf{r_i(x_j)} = \mathit{max}~\{ \pdvOf{w_l(y_l)} ~|~  w_l(y_l) \leq_{h}^{\mathcal{P}} r_i(x_j) \\
      ~~~~~~~~~~~~~~~~~~~~~~~~~~~~~~~~~~~~~~~~~ \land \left( \forall k : x_j \versionOrder_h x_k \implies w_l(y_l) \leq_{h}^{\mathcal{P}} w_k(x_k) \right) \} \\
      \pdvOf{w_i(x_i)} = \mathit{max}~\{ \pdvOf{r_i(y_j)} ~|~ y_j \in \readSetOf{T_i} \}~\union \\
      ~~~~~~~~~~~~~~~~~~~~~~~~ \{ \pdvOf{w_k(z_k)} : w_k(z_k) \leq_{h}^{\mathcal{P}} w_i(x_i)  \} \\
      ~~~~~~~~~~~~~~~~~ + ~ \Sigma_{X \in \mathcal{P}(\writeSetOf{T_i})} ~1_X
    \end{array}
  \end{displaymath}
\end{definition}

The first two rules of function \pdv are identical to 
the ones that would give us function \dv on history $h^{\mathcal{P}}$.
The second part of the third rule serializes objects in the same partition 

We now prove that partitioned dependence vectors properly capture consistent snapshots.
Consider the following definition of \isCompatible{T_i,x_j,y_l} for a proper partition $\mathcal{P}$:
\begin{itemize}
\item[\textbf{Case $\mathcal{P}(x) \neq \mathcal{P}(y)$.}]
  This case is identical to the definition we gave for function \dv.
  In other words, both $\pdvOf{r_i(x_j)}[\mathcal{P}(x)] \geq \pdvOf{r_i(y_l)}[\mathcal{P}(x)]$
 and $\pdvOf{r_i(y_l)}[\mathcal{P}(y)] \geq \pdvOf{r_i(x_j)}[\mathcal{P}(y)]$ must hold.
\item[\textbf{Case $\mathcal{P}(x) = \mathcal{P}(y)$.}]
  This case deals with the fact that inside a partition writes are serialized.
  We have
  \begin{inparaenum}[(i)]
  \item if $\pdvOf{r_i(x_j)}[\mathcal{P}(y)] > \pdvOf{r_i(y_l)}[\mathcal{P}(y)]$ holds 
    then $y_l=\mathit{max}~\{ y_k ~|~ w_k(y_k) \leq_h^{\mathcal{P}} w_j(x_j) \}$, or symmetrically
  \item if $\pdvOf{r_i(y_l)}[\mathcal{P}(x)] > \pdvOf{r_i(x_j)}[\mathcal{P}(x)]$ holds
    then $x_j=\mathit{max}~\{ x_k ~|~ w_k(x_k) \leq_h^{\mathcal{P}} w_l(y_l) \}$, or otherwise
  \item the predicate equals \true.
  \end{inparaenum}
\end{itemize}

We prove next that the ``if'' part of \reftheo{protocol:1} holds
for the above definition of compatibility:

\begin{proposition}
  \labprop{protocol:2} 
  Consider a  history $h$ in \NMSI and a transaction $T_i$ in $h$.
  If every pair of versions $x_j$ and $y_l$ read by $T_i$ is compatible,
  then transaction $T_i$ sees a consistent snapshot in $h$
\end{proposition}

\begin{proof}
  Using a reasoning identical to the one we depicted in the proof of \reftheo{protocol:1}, 
  we can prove that $h^{\mathcal{P}}$ belongs to \CONS.
  Then, from \refprop{protocol:1}, we know that if $h^{\mathcal{P}}$ belongs to \CONS, then $h$ belong to \CONS.
\end{proof}

As discussed in \cite{psutra-hotcdp12}, we notice here the existence of 
a trade-off between the size of the vectors and the freshness of the snapshots. 
For instance, 
if $x$ and $y$ belong to the same partition and transaction $T_i$ reads a version $x_j$, 
$T_i$ cannot read a version $y_l$ that committed after a version $x_k$ posterior to $x_j$.

\subsection{Transaction Lifetime in \jessy}
\labsection{protocol:lifetime}

\jessy is a distributed system of processes which communicate by message passing.
When a client (not modeled) executes a transaction $T_i$ with \jessy, 
$T_i$ is handled by a coordinator.
The coordinator of a transaction can be any process in the system.
A transaction $T_i$ can be in one of the following four states at some process:
\begin{compactitem}

\item \transExecuting:
   Each non-termination operation $o_i(x)$ in $T_i$ is executed
   optimistically (i.e., without synchronization with other replicas) at
   the transaction coordinator \dbCoordinatorOf{T_i}.
   If $o_i(x)$ is a read, \dbCoordinatorOf{T_i} returns the corresponding value,
   fetched either from the local replica or a remote one.
   If $o_i(x)$ is a write, \dbCoordinatorOf{T_i} stores the corresponding update
   value in a local buffer, enabling \emph{(i)} subsequent reads to observe the
   modification, and \emph{(ii)} a subsequent commit to send the write-set to
   remote replicas.

\item \transSubmitted:
  Once all the read and write operations of $T_i$ have executed, 
  $T_i$ terminates, and the coordinator submits it to the termination protocol.
  The protocol applies a certification test on $T_i$ to enforce \NMSI. 
  This test ensures that if two concurrent conflicting update transactions
  terminate, one of them aborts.
 
\item \transCommitted/\transAborted:
  When $T_i$ enters the \transCommitted state at $r \in \replicaSetOf{T_i}$,
  its updates (if any) are applied to the local data store.
  If $T_i$ aborts, $T_i$ enters the \transAborted state.
  
\end{compactitem}

\subsection{Execution Protocol} 
\labsection{protocol:execution}

\begin{algorithm}[t]
  \footnotesize
  \caption{Execution Protocol of \jessy}
  \labalg{execution}

    \begin{algorithmic}[1]
    
    \StartVariables
    \VarCustom{\dbDatabase, \dbSubmitted, \dbCommitted, \dbAborted}
    \EndVariables

    \StartAction{$\dbRemoteRead{x,T_i}$} 
    \Precondition{$\receivedFrom{\flagReadRequest,T_i,x}{q}$} \labline{em:1}
    \Precondition{$\exists (x,v,j) \in \dbDatabase : \forall y_l \in \readSetOf{T_i} : \isCompatible{T_i,x_j,y_l}$} \labline{em:2}
    \Effect{$\send{\flagReadReply,T_i,x,v}{q}$} \labline{em:5}
    \EndAction

    \StartAction{$\dbExecute{\flagWrite,x,v,T_i}$} 
    \Effect{$\updateSetOf{T_i} \assign \updateSetOf{T_i} \union \{(x,v,i)\}$} \labline{em:6}
    \EndAction

    \StartAction{$\dbExecute{\flagRead,x,T_i}$}
    \Effect{\textbf{if} $\exists (x,v,i) \in \updateSetOf{T_i}$ \textbf{then} \textbf{return} $v$} \labline{em:7}
    \Effect{\textbf{else}} \labline{em:8}
    \Effect{\hspace{1em}$\send{\flagReadRequest,T_i,x}{\replicaSetOf{x}}$} \labline{em:9}
    \Effect{\hspace{1em}\textbf{wait until} $\receivedAny{\flagReadReply,T_i,x,v}$} \labline{em:10}
    \Effect{\hspace{1em}\textbf{return} $v$} \labline{em:11}
    \EndAction

    \StartAction{$\dbExecute{\flagTerminate,T_i}$}
    \Effect{$\dbSubmitted \assign \dbSubmitted \union \{T_i\}$} \labline{em:12}
    \Effect{\textbf{wait until} $T_i \in \dbDecided$} \labline{em:13}
    \Effect{\textbf{if} $T_i \in \dbCommitted$ \textbf{then} \textbf{return} $\flagCommit$} \labline{em:14}
    \Effect{\textbf{return} $\flagAbort$} \labline{em:16}
    \EndAction

  \end{algorithmic}

\end{algorithm}

\refalg{execution} describes the execution protocol in pseudocode.
Logically, it can be divided into two parts:
action \dbRemoteRead{}, executed at some process, 
reads an object replicated at that process in a consistent snapshot; and
the coordinator $\coordOf{T_i}$ performs actions
\dbExecute{} to execute $T_i$ and to buffer the updates in \updateSetOf{T_i}.

The variables of the execution protocol are:
\dbDatabase, the local data store;
\dbSubmitted contains locally-submitted transactions; and
\dbCommitted (respectively \dbAborted) stores committed (respectively aborted) transactions.
We use the shorthand \dbDecided for $\dbCommitted \union \dbAborted$.

Upon a read request for $x$, \dbCoordinatorOf{T_i} checks against \updateSetOf{T_i} if
$x$ has been previously updated by the same transaction; if so, it
returns the corresponding value (\refline{em:7}).
Otherwise, \dbCoordinatorOf{T_i} sends an (asynchronous) read request
to the processes that replicate $x$ (\reflines{em:9}{em:11}).
When a process receives a read request for object $x$ that it
replicates, it returns a version of $x$ which complies with  \reftheo{protocol:1}
(\reflines{em:1}{em:5}).

Upon a write request of $T_i$, the process buffers the update value in \updateSetOf{T_i} (\refline{em:6}).
During commitment, the updates of $T_i$ will be sent to all
replicas holding an object that is modified by $T_i$ .

When transaction $T_i$ terminates, it is submitted to the termination protocol
(\refline{em:12}). 
The execution protocol then waits until $T_i$ either commits or
aborts, and returns the outcome.

\subsection{Termination Protocol}
\labsection{protocol:termination} 

\refalg{termination} depicts the termination protocol of \jessy.
It accesses the same four variables \dbDatabase, \dbSubmitted and \dbCommitted, 
along with a FIFO queue named \dbCertifyQueue.

In order to satisfy \GPR, the termination protocol uses a genuine atomic
multicast primitive \cite{GUERRAOUI2001}.
In our model, this requires that either 
\begin{inparaenum}[(i)]
\item we form non-intersecting groups of replicas, and an eventual leader oracle is available in each group, or
\item that a system-wide \emph{reliable} failure detector is available.
\end{inparaenum}
The latter setting allows \jessy to tolerate a disaster \cite{nicolasPHD}.

To terminate an update transaction $T_i$,
\dbCoordinatorOf{T_i} atomic-multicasts it to every process that holds an object written by $T_i$.
Every such process $p$ certifies $T_i$ by calling function \dbCertify{T_i}
(\refline{tp:7}).
This function returns \true at process $p$, 
iff for every transaction $T_j$ committed prior to $T_i$ at $p$,
if $T_j$ write-conflicts with $T_i$, then $T_i$ depends on $T_j$.
Formally:
\begin{displaymath}
  \dbCertify{T_i} \equaldef \forall T_j \in \dbCommitted : \writeSetOf{T_i} \inter \writeSetOf{T_j} \neq \emptySet \implies T_i \depend T_j 
\end{displaymath}

Under partial replication, a process $p$ might store only a subset of the
objects written by $T_i$, in which case $p$ does not have enough information to decide on the outcome of $T_i$.
Therefore, we introduce a voting phase where replicas of the objects written by
$T_i$ send the result of their certification test in a \flagVote message
to every process in $\wreplicaSetOf{T_i} \union \{\coordOf{T_i}\}$ (\reflines{tp:8a}{tp:8b}).

A process can safely decide on the outcome of $T_i$ when it has received votes from a \emph{voting quorum} for $T_i$.
A voting quorum $Q$ for $T_i$ is a
set of replicas such that for every object $x \in \dbCertifiedSetOf{T_i}$,
the set $Q$ contains at least one of the processes replicating $x$. 
Formally, a set of processes is a voting quorum for $T_i$ iff it belongs to \votingQuorum{T_i}, defined as follows:
\begin{displaymath}
  \votingQuorum{T_i} \equaldef  \{ Q \subseteq \Pi ~\st~ \forall x \in \dbCertifiedSetOf{T_i} : \exists j \in Q \inter \replicaSetOf{x} \}
\end{displaymath}

A process $p$ makes use of the following (three-values) predicate $\dbVoteOutcome{T_i}$
to determine whether some transaction $T_i$ commits, or not:
\vspace{-0.3em}
\begin{eqnarray*}
  \dbVoteOutcome{T_i} \equaldef \\
  &\begin{array}{r@{\hspace{0.4em}}l@{\hspace{-1em}}l}
    \multicolumn{3}{l}{\textbf{if}~\dbCertifiedSetOf{T_i}=\emptySet} \\
    \multicolumn{3}{l}{~~\textbf{then}~\true} \\
    \textbf{else} & \textbf{if}   & \forall Q \in \votingQuorum{T_i}, \exists q \in Q, \\
                   &               & ~~ \neg\receivedFrom{\flagVote,T,-}{q} \\ 
                  &
                  \multicolumn{2}{l}{~~\textbf{then}~\dbUnknown}
                  \\
    \textbf{else} & \textbf{if}   & \exists Q \in \votingQuorum{T_i}, \forall q \in Q, \\
                  &               & ~~ \receivedFrom{\flagVote,T,\true}{q} \\
                  & \multicolumn{2}{l}{~~\textbf{then}~\true} \\
    \textbf{else} & \false
   \end{array}
\end{eqnarray*}

To commit transaction $T_i$, process $p$ first applies $T_i$'s updates to its local data store,
then $p$ adds $T_i$ to variable \dbCommitted (\reflines{tp:9}{tp:12}).
If instead $T_i$ aborts, $p$ adds $T_i$ to \dbAborted (\reflines{tp:13}{tp:14}).

\begin{algorithm}[t]

  \footnotesize
  \caption{Termination Protocol of \jessy}
  \labalg{termination}

  \begin{algorithmic}[1]
    
    \StartVariables
    \VarCustom{\dbDatabase, \dbSubmitted, \dbCommitted, \dbAborted, \dbCertifyQueue}
    \EndVariables
   
    \StartAction{$\dbSubmit{T_i}$}
    \Precondition{$T_i \in \dbSubmitted$} \labline{tp:0}
    \Precondition{$\writeSetOf{T_i} \neq \emptySet$} \labline{tp:1}
    \Effect{$\amcast{T_i}{\wreplicaSetOf{T_i}}$} \labline{tp:2}
    \EndAction    

    \StartAction{$\dbDeliver{T_i}$}
    \Precondition{$T_i=\amdeliver$} \labline{tp:3}
    \Effect{$\dbCertifyQueue \assign \dbCertifyQueue \seqAppend \langle T_i \rangle$} \labline{tp:4}
    \EndAction

    \StartAction{$\dbVote{T_i}$}
    \Precondition{$T_i \in \dbCertifyQueue \setminus \dbDecided$} \labline{tp:5}
    \Precondition{$\forall T_j \in \dbCertifyQueue,~T_j <_{\dbCertifyQueue} T_i \implies T_j \in \dbDecided$} \labline{tp:6}
    \Effect{$v \assign \dbCertify{T_i}$} \labline{tp:7}
    \Effect{$\send{\flagVote,T_i,v}{\wreplicaSetOf{T_i}}$} \labline{tp:8a}
    \Effect{\hspace{10em}$\union~\{\coordOf{T_i}\}$} \labline{tp:8b}
    \EndAction

    \StartAction{$\dbCommit{T_i}$}
    \Precondition{$\dbVoteOutcome{T_i}$} \labline{tp:9}
    \Effect{\textbf{foreach} $(x,v,i)$ \textbf{in} $\updateSetOf{T_i}$
    \textbf{do}} \labline{tp:10}
    \Effect{\hspace{1em}\textbf{if} $x \in \dbDatabase$ \textbf{then}
                                       $\dbDatabase \assign \dbDatabase \union
                                       \{(x,v,i)\}$} \labline{tp:11}
    \Effect{$\dbCommitted \assign \dbCommitted \union \{T_i\}$} \labline{tp:12}
    \EndAction

    \StartAction{$\dbAbort{T_i}$}
    \Precondition{$\neg \dbVoteOutcome{T_i}$} \labline{tp:13}
    \Effect{$\dbAborted \assign \dbAborted \union \{T_i\}$} \labline{tp:14}
    \EndAction

  \end{algorithmic}
\end{algorithm}

\subsection{Correctness of \jessy}
\labappendix{jessy}

We now sketch a correctness proof of \jessy:
\refprop{cj:1} establishes that \jessy generates histories in \NMSI.
\refprop{cj:2} shows that read-only transactions are wait-free.
Propositions~\ref{prop:cj:3} and \ref{prop:cj:4}, respectively, 
prove that \jessy satisfies obstruction-freedom for updates and non-triviality for \NMSI.

\subsubsection{Safety}
\labappendix{jessy:safety}

\begin{proposition}
  \labprop{cj:0}
  If a transaction $T_i$ commits (respectively aborts) at some process in $\wreplicaSetOf{T_i} \union \coordOf{T_i}$, 
  it commits (resp. aborts) at every correct process in $\wreplicaSetOf{T_i} \union \coordOf{T_i}$.
\end{proposition}

\begin{proof}
  This proposition follows from
  the properties of atomic multicast,
  the fact that the queue \dbCertifyQueue is FIFO,
  the preconditions at \reflines{tp:5}{tp:6} in \refalg{termination},
  and the definitions of \dbVote{} and \dbVoteOutcome{}.
\end{proof}

\begin{proposition}
  \labprop{cj:1}
  Every history admissible by \jessy belongs to \NMSI.
\end{proposition}

\begin{IEEEproof}  
  We first observe that transactions in \jessy always read committed versions of the objects (\refline{em:2} in \refalg{execution}).
  Moreover, we know by \reftheo{protocol:1} that reads are consistent when \jessy uses dependence vectors,
  and that this property also holds in case \jessy employs partitioned dependence vectors (\refprop{protocol:2}).
  It thus remains to show that histories generated by \jessy are write-conflict free (\WCF).

  To prove that \WCF holds, we consider two independent write-conflicting transactions $T_i$ and $T_j$,
  and we assume for the sake of contradiction that they both commit.
  We note $p_i$ (resp. $p_j$) the coordinator of $T_i$ (resp. $T_j$).
  Since $T_i$ and $T_j$ write-conflict, 
  there exists some object $x$ in $\writeSetOf{T_i} \inter \writeSetOf{T_j}$.
  One can show that the following claim holds:
  \begin{itemize}
  \item[(C1)]
    For any two replicas $p$ and $q$ of $x$, 
    denoting $\dbCommitted_p$ (resp. $\dbCommitted_q$) the set $\{ T_j \in \dbCommitted: x \in \writeSetOf{T_j} \}$,
    at the time $p$ (resp. $q$) decides $T_i$,
    it is true that $\dbCommitted_p$ equals $\dbCommitted_q$.
  \end{itemize}

  According to \refline{tp:9} of \refalg{termination} and the definition of function \dbVoteOutcome{},
  $p_i$ (respectively $p_j$) received a positive \flagVote message 
  from some process $q_i$ (resp. $q_j$) replicating $x$.
  Observe that $T_i$ (resp. $T_j$) is in variable \dbCertifyQueue at process $q_i$ (resp. $q_j$)
  before this process sends its \flagVote message.
  It follows from claim C1 that either (1) at the time $q_i$ sends its \flagVote message, $T_j <_{\dbCertifyQueue} T_i$ holds,
  or (2) at the time $q_j$ sends its \flagVote message, $T_i <_{\dbCertifyQueue} T_j$ holds.
  Assume that case (1) holds (the reasoning for case (2) is symmetrical).
  From the precondition at \refline{tp:6} in \refalg{termination},
  we know that process $q_i$ must wait that $T_j$ is decided before casting a vote for $T_i$.
  From \refprop{cj:0}, we deduce that $T_j$ is committed at process $q_i$.
  Hence, \dbCertify{T_i} returns \false at process $q_i$; a contradiction.
\end{IEEEproof}

\subsubsection{Progress}
\labappendix{jessy:progress}

\begin{lemma}
  \lablem{cj:2}
  For every transaction $T_i$,
  if \coordOf{T_i} is correct,
  eventually $T_i$ is submitted to the termination protocol at \coordOf{T_i}.
\end{lemma}

\begin{IEEEproof}
  Transaction $T_i$ executes all its write operations locally at its coordinator.
  Now, upon executing a read request on some object $x$, if $x$ was modified previously
  by $T_i$, the corresponding value is returned.
  Otherwise, \coordOf{T_i} sends a read request to \replicaSetOf{x}. 
  To prove this lemma,  
  we have to show that eventually one of the replica replies to the coordinator.

  According to our model, there exists one correct process replica of $x$.
  In what follows, we name it $p$.
  Observe that since links are quasi-reliable, $p$ eventually receives the read request from \coordOf{T_i}.
  Upon receiving this request,
  process $p$ tries returning a version of $x$ compatible with all versions previously read by $T_i$.

  Consider that \jessy uses dependence vectors (the reasoning for partitioned dependence vectors is similar),
  and assume, by contradiction, that $p$ never finds such a compatible version.
  From the definition of \isCompatible{T_i,x_j,y_l}, 
  this means that the following predicate is always true:
  \begin{displaymath}
    \begin{array}{l@{~}l}
      \forall (x,v,l) \in \dbDatabase: & \dvOf{w_l(x_l)}[x] < \dvOf{r_i(y_j)}[x] \\
                                       & \lor~ \dvOf{w_l(x_l)}[y] > \dvOf{r_i(y_j)}[y]
    \end{array}
  \end{displaymath}
  This means that there exists a version $x_k$ upon which transaction $T_i$ depends,
  and such that $\dvOf{w_k(x_k)}[x]=\dvOf{r_i(y_j)}[x]$.
  Transaction $T_k$ committed at some site.
  As a consequence, \refprop{cj:0} tells us that eventually $T_k$ commits at process $p$.
  We conclude by observing that since \jessy satisfies both \CONS and \WCF,
  $\dvOf{w_k(x_k)}[y] > \dvOf{r_i(y_j)}[y]$ cannot hold.

\end{IEEEproof}

\begin{lemma}
  \lablem{cj:1}
  For every transaction $T_i$,
  if $T_i$ is submitted at \coordOf{T_i} and \coordOf{T_i} is correct,
  every correct process in $\wreplicaSetOf{T_i} \union \coordOf{T_i}$ eventually decides $T_i$.
\end{lemma}

\begin{IEEEproof}
  According to \reflem{cj:2} and the properties of atomic multicast, transaction $T_i$ is delivered 
  at every correct process in $\wreplicaSetOf{T_i} \union \coordOf{T_i}$.
  It is then enqueued in variable \dbCertifyQueue (\reflines{tp:3}{tp:4} in \refalg{termination}).

  Because \dbCertifyQueue is FIFO, 
  processes dequeue transactions in the order they deliver them (\reflines{tp:5}{tp:6}).
  The uniform prefix order and acyclicity properties of genuine atomic multicast ensure
  that no two processes in the system wait for a vote from each other.
  It follows that every correct replica in \wreplicaSetOf{T_i} eventually dequeues $T_i$, and sends
  the outcome of function \dbCertify{T_i} to $\wreplicaSetOf{T_i} \union \coordOf{T_i}$ (\reflines{tp:7}{tp:8b}).

  Since there exists at least one correct replica for each object modified by $T_i$
  eventually every correct process in $\wreplicaSetOf{T_i} \union \coordOf{T_i}$ collects enough votes to decide upon the outcome of $T_i$.
\end{IEEEproof}

\begin{proposition}
  \labprop{cj:2}
  \jessy satisfies \WFQ.
\end{proposition}

\begin{IEEEproof}
  Consider some read-only transaction $T_i$ and assume that \coordOf{T_i} is correct,
  \reflem{cj:2} tells us that $T_i$ is eventually submitted at \coordOf{T_i}.
  According to the definition of predicate \dbVoteOutcomeFunction, \dbVoteOutcome{T_i} always equals true. 
  Hence, the precondition at \refline{tp:9} in \refalg{termination} is always true,
  whereas precondition at \refline{tp:13} is always false.
  It follows that $T_i$  eventually commits.
\end{IEEEproof}

We now prove that \jessy satisfies obstruction-freedom for updates (\OFU) and non-triviality for \NMSI.
These results are both stated in the case where \jessy employs non-partitioned dependence vectors.
The question of ensuring any of these properties with a smaller space-complexity than $O(m)$ 
where $m$ is the number of objects in the system remains open.

\begin{proposition}
  \labprop{cj:3}
  Jessy ensures non-trivial \NMSI.
\end{proposition}

\begin{IEEEproof}
  Consider a replica $p$ of $x$ storing version $x_j$,
  and assume an extension of the execution in which 
  $p$ answers first to a remote read request from \coordOf{T_i} over $x$.
  Since history $h.r_i(x_j)$ is in $\NMSI$, it belongs to $\CONS$. 
  Because \jessy use dependence vectors, \reftheo{protocol:1} tells us that:
  $\dvOf{r_i(x_j)}[x] \geq \dvOf{r_i(y_k)}[x]$
  and $\dvOf{r_i(x_j)}[y] \leq \dvOf{r_i(y_k)}[y]$ hold.
  According to the preconditions of operation $\dbRemoteRead{x,T_i}$ and modification M1, 
  process $p$ returns version $x_j$ to \coordOf{T_i}.
\end{IEEEproof}

\begin{proposition}
  \labprop{cj:4}
  \jessy satisfies \OFU.
\end{proposition}

\begin{IEEEproof}

  Consider an execution $\run$ of \jessy and note $h=\refMapOf{\run}$ the history produced by $\run$.
  Let $T_i$ be an update transaction not executed in $\run$.
  First of all, we observe that in any continuation of $\run$ during which $\coordOf{T_i}$ is correct,
  from \reflem{cj:1}, \coordOf{T_i} eventually decides transaction $T_i$.
  Then, assume that $T_i$ is not conflicting 
  in some continuation $h'=\refMapOf{\run  \sqsubseteq \run'}$ with any concurrent transaction in $h'$.
  This means that for every transaction $T_j$, if $T_j$ conflicts with $T_i$, 
  then $T_i$ depends upon $T_j$.
  Accordingly to \reftheo{protocol:1}, 
  the code at \refline{tp:7} in \refalg{termination}, 
  and the definition of function \dbCertify{}, 
  transaction $T_i$ commits in $h'$.
  
\end{IEEEproof}


\newpage 
\section{Related Work}
\labsection{relatedwork} 
 
\reftab{comparisonTable} compares different partial replication protocols, 
in terms of time and message complexity (from the coordinator's perspective), 
when executing a transaction with $r_r$ remote reads and $w_r$ remote writes.
A transaction can be of the following three types:
a read-only transaction,
a local update transaction (the coordinator replicates all the objects accessed by the transaction), 
or a global update transaction (some object is not available at the coordinator).

Several protocols solve particular instances of the partial replication problem.
Some assume that a correct replica holds all the data accessed by a transaction \cite{CPV05,Schiper2006} .
Others consider that data can be partitioned into conflict sets  \cite{icdcs02-kemme},
or that always aborting concurrent conflicting transactions \cite{Holliday2002} is reasonable.
Hereafter, we review in details algorithms that do not make such an assumption.

P-Store \cite{Schiper2010} is a genuine partial replication algorithm that ensures \SER by leveraging genuine atomic multicast.
Like in \jessy, read operations are performed optimistically at some replicas and update operations are applied at commit time.
However, unlike Jessy, P-Store certifies read-only transactions as well.

A few algorithms \cite{Serrano2007,Armendariz-Inigo2008} offer partial replication with \SI semantics.
At the start of a transaction $T_i$, 
the algorithm of \citet{Armendariz-Inigo2008} atomically broadcasts $T_i$ to all processes.
This message defines the consistent snapshot of $T_i$.
If $T_i$ is an update transaction,
$T_i$'s write set is atomic broadcast to all processes at commit time and each process independently certifies it.
The algorithm of \citet{Serrano2007} executes a dummy transaction after each commit.
As the commit of a transaction is known by all processes, a dummy transaction identifies a snapshot point.
This avoids the cost of the start message.
As a consequence of the impossibility result depicted in \refsection{imp}, 
none of these algorithms is genuine.

Walter is a transactional key-value store proposed by \citet{Sovran2011}
that supports Parallel Snapshot Isolation (\PSI). 
PSI is somewhat similar to \NMSI; in particular, \PSI snapshots are non-monotonic.
However, \PSI is stronger than \NMSI, as it enforces \SCONSa:
\NMSI allows reading versions of objects that have committed after the start of the
transaction, as long as it is consistent.
On the contrary in PSI, 
an operation has to read the most recent versions at the time the transaction starts.
Enforcing \SCONSa does not preclude any anomaly,
and it increases the probability that a write skew, or a conflict between concurrent writes occurs.
To ensure \PSI, Walter relies on a single master replication schema per object and 2PC.
After the transaction commits, it is propagated to all processes in the system in
the background before it becomes visible.

More recently, \citet{GMU} proposed GMU, an algorithm 
that supports an extended form of update serializability.
GMU relies on vector clocks to read consistent snapshots.
At commit time, both GMU and Walter use locks to commit transactions.
Because locks are not ordered before voting (contrary to P-Store and Jessy),
these algorithms are subjected to the occurrence of distributed deadlocks, and scalability problems leading to
poor performance for global update transactions \cite{Gray1996,Wiesmann2005}.



\begin{table*}[t]
\footnotesize
\begin{center}
  \hspace*{-3em}
  \begin{tabular}{@{}c | c  c  c | c | c  | c |  c@{}}
    \multicolumn{5}{c}{} & \multicolumn{3}{c}{Time complexity}
    \\ \cline{6-8} 
    \multicolumn{1}{c|}{Algorithm} & \multicolumn{1}{p{8mm}}{\centering Cons.} & \multicolumn{1}{p{9mm}}{\centering Gen\-u\-ine?} & \multicolumn{1}{p{10mm}}{\centering Multi-Master?} & \multicolumn{1}{@{}p{19mm}@{}|}{\centering Message Complexity} & \multicolumn{1}{c}{Read-only} & \multicolumn{1}{c}{Global Update} & \multicolumn{1}{p{12mm}}{\centering Local Update}\\
    \hline
    
    \cline{7-8}
    P-Store \cite{Schiper2010} & \SER & yes & yes & $O(n^2)$ & $(r_r \times 2\Delta) + 4\Delta$ & $(r_r \times 2\Delta) + 5\Delta$ & $4\Delta$ \\

    GMU \cite{GMU} & \US & yes & yes & $O(n^2)$ & $r_r \times 2\Delta$ & $(r_r \times 2\Delta) + 2\Delta$ & $2\Delta$ \\

    SIPRe\cite{Armendariz-Inigo2008} & \SI & no  &yes & $O(N^2)$ & $(r_r \times 2\Delta) + 3\Delta$ & $ (r_r + w_r) \times 2\Delta + 6\Delta$& $6\Delta$ \\
    Serrano\cite{Serrano2007} & \SI & no & yes & $O(N^2)$ & $r_r \times 2\Delta$ & $ (r_r + w_r) \times 2\Delta + 3\Delta$ & $3\Delta$ \\
    Walter \cite{Sovran2011} & \PSI & no & no & $O(N)$ & $r_r \times 2\Delta$ &  $(r_r \times 2\Delta) + 2\Delta$ & $2\Delta\ |\ 0$ \\
    \jessy & \NMSI & yes & yes & $O(w_r{^2})$ & $r_r \times 2\Delta$ & $(r_r \times
    2\Delta) + 5\Delta$ & $4\Delta$ \\
  \end{tabular}

\end{center}
{\footnotesize\em
  Message complexity: number of messages sent on behalf of transaction.  
  Time complexity: delay for  executing a transaction.
  $N$: number of replicas;
  $n$: number of replicas involved in transaction;
  $\Delta$: message latency between replicas;
  $r_r$: number of remote reads;
  $w_r$: number of remote writes.
  The latency of atomic broadcast (resp. atomic multicast) is considered
  $3\Delta$ (resp  $4\Delta$) during solo step execution \cite{nicolasPHD}.
}

\caption{Comparison of partial replication protocols}
\labtab{comparisonTable}
\end{table*}

\newpage
\section{Conclusion}
\labsection{conclusion}

Partial replication and genuineness are two key factors of scalability in replicated systems.
This paper shows that ensuring snapshot isolation (\SI) in a genuine partial replication (\GPR) system is impossible.
To state this impossibility result, we introduce four properties whose conjunction is equivalent to \SI.
We show that two of them, namely snapshot monotonicity and strictly consistent snapshots cannot be ensured.

To side step the incompatibility of \SI with \GPR, we propose a novel consistency criterion named \NMSI.
\NMSI prunes most anomalies disallowed by \SI, while providing guarantees close to \SI:
transactions under \NMSI always observe consistent snapshots and two write-conflicting concurrent updates never both commit.

The last contribution of this paper is \jessy, a genuine partial replication protocol that supports \NMSI.
To read consistent partial snapshots of the system,
\jessy uses a novel variation of version vectors called dependence vectors.
An analytical comparison between \jessy and previous partial replication protocol
shows that \jessy contacts fewer replicas, and that, in addition, it may commit faster.

\section*{Acknowledgments}
We thank Sameh Elnikety and Vivien Qu\'{e}ma for insightful discussions and
feedbacks.

\newpage
 
\bibliographystyle{IEEEtranN} 
\bibliography{IEEEabrv,library,bib,my,nicolas,psutra}
  
\end{document}